\documentclass[a4paper,twocolumn]{article}

\usepackage{amsmath,amssymb}

\usepackage{anyfontsize}

\usepackage[utf8]{inputenc}
\usepackage[T1]{fontenc}

\usepackage[sfdefault]{biolinum}
\usepackage[sansmath]{libertinust1math}
\usepackage{stmaryrd}

\allowdisplaybreaks

\usepackage[a4paper, left=.5in,right=.5in,top=1in,bottom=1in]{geometry}

\usepackage{cite}
\usepackage{algorithmic}
\usepackage{algorithm}
\usepackage{graphicx}
\usepackage{braket}
\usepackage{url}
\usepackage{booktabs}
\usepackage[justification=centering]{caption}
\usepackage{subcaption}
\usepackage{xcolor} 
\usepackage{diagbox}
\usepackage{multirow}
\usepackage{makecell}
\usepackage{rotating}

\usepackage{hyperref}
\usepackage{cleveref}
\hypersetup{
  colorlinks=true,
  linkcolor=blue,
  filecolor=magenta,      
  urlcolor=cyan
}

\usepackage{flushend}

\usepackage[english]{babel}

\newcommand{\ceil}[1]{\lceil{#1}\rceil}
\newcommand{\floor}[1]{\lfloor{#1}\rfloor}

\usepackage{amsthm}

\newtheorem{theorem}{Theorem}[section]
\newtheorem{corollary}{Corollary}[theorem]
\newtheorem{lemma}[theorem]{Lemma}

\theoremstyle{definition}
\newtheorem{problem}{Problem}[theorem]

\title{Reducing the Depth of Linear Reversible Quantum
  Circuits\footnote{The published version can be found in \textit{IEEE
      Transactions on Quantum Engineering}, vol. 2, pp. 1-22, 2021,
    Art no. 3102422, 
\href{https://dx.doi.org/10.1109/TQE.2021.3091648}{doi:10.1109/TQE.2021.3091648}.}}
\author{
Timothée Goubault de Brugière$^{1,3}$,
Marc Baboulin$^{1}$, Benoît Valiron$^{2}$,\\
Simon Martiel$^{3}$ and Cyril Allouche$^{3}$
}
\date{%
  $^{1}${\textit{Laboratoire de Recherche en Informatique, Université Paris-Saclay}, 
Orsay, France}\\
$^2${\textit{Laboratoire de Recherche en Informatique, CentraleSupélec}, 
Orsay, France}\\
$^3${\textit{Atos Quantum Lab}, Les Clayes-sous-Bois, France}}

\begin{document}

\maketitle

\begin{abstract}
In quantum computing the decoherence time of the qubits determines the computation time available and this time is very limited when using current hardware. In this paper we minimize the execution time (the depth) for a class of circuits referred to as linear reversible circuits, which has many applications in quantum computing (e.g., stabilizer circuits, ``CNOT+T'' circuits, etc.). We propose a practical formulation of a divide and conquer algorithm that produces quantum circuits that are twice as shallow as those produced by existing algorithms. We improve the theoretical upper bound of the depth in the worst case for some range of qubits. We also propose greedy algorithms based on cost minimization to find more optimal circuits for small or simple operators. Overall, we manage to consistently reduce the total depth of a class of reversible functions, with up to $92\%$ savings in an ancilla-free case and up to $99\%$ when ancillary qubits are available. 
\end{abstract}

\section{Introduction}

Quantum computing is getting closer to the moment when it will be able to solve problems
insoluble using current computers. The manipulation of qubits is increasingly controlled, quantum gates
are performed with better fidelity and works for achieving quantum supremacy have been proposed \cite{arute2019quantum,zhong2020quantum}, even though the significance of such works remains highly debated \cite{ibm_blog, aaronson_blog}.

In addition to the noise inherent in manipulating qubits, there is another phenomenon to control: quantum decoherence. The qubits must remain isolated from the outside world during the execution of the algorithm or else they may interact unintentionally
with external elements which would distort the results. It is still difficult to isolate these qubits for a long time.
If a hardware improvement is possible, it is also possible to compress the set of instructions so that their execution
takes less time. These instructions are usually represented by a quantum circuit and, assuming that two non-overlapping gates can be executed in parallel, the execution time of the circuit is strongly related to its depth. Thus, the proper execution of complex algorithms can be significantly facilitated by optimizing the depth of quantum circuits.

In a fault-tolerant computational model, the T gate is the most expensive gate to implement \cite{bravyi2005universal}. As a consequence, during these last years a lot of efforts have been made to minimize this resource, whether it is the T-count \cite{2058-9565-4-1-015004,amy2019t,kissinger2020reducing,zhang2019optimizing} or the T-depth \cite{DBLP:journals/tcad/AmyMM14}. However, these optimizations often come with an increase of other resources, especially the number of CNOTs or the total depth of the circuit. Such an additional cost is not negligible and can affect the final outputs of a quantum algorithm, see \cite{maslov2016optimal} for more details about why the cost of a quantum circuit should not be reduced to the T-cost. It is therefore important to minimize the secondary resources as well, if possible while keeping the T-cost unchanged.

Work in this direction has been carried out recently. It has mainly focused on the number of CNOTs but in a NISQ (Noisy Intermediate-Scale Quantum) setting. Overall, a significant decrease of the CNOT count is reported but with an increase of the T-depth \cite{amy2018controlled}. With NISQ computers the T-depth is less important to optimize because the T-gate is not implemented fault-tolerantly, but for fault-tolerant computations we have to find a better compromise between the T-depth and the CNOT cost. Our goal in this paper is to achieve this better compromise by improving the depth of quantum circuits while keeping the T-depth as low as possible.

For this, we are interested in the optimization of a subclass of circuits called ``linear reversible circuits''. 
These circuits can be rewritten with only CNOT gates and have direct applications in other more 
complex circuit structures such as stabilizer circuits or CNOT+T circuits, two classes of circuits that have shown crucial utility in the design of efficient
quantum compilers \cite{2058-9565-4-1-015004,DBLP:journals/tcad/AmyMM14} and error
correcting codes \cite{gottesman1997stabilizer}. Hence, the synthesis of CNOT circuits occurs naturally in general quantum compilers and can be used as a first approach for optimizing general circuits.

In this article we present two kinds of algorithms for the synthesis of linear reversible circuits. Our contributions are the following : 
\begin{itemize}
\item We present DaCSynth, a practical implementation of a divide and conquer framework that divides
  the synthesis into parallelizable sub-problems that can be solved
  with several strategies.
\item We give strict upper bounds on the depth of the circuits. First
  we prove that, in all generality, the depth is upper bounded by
  $2n +2\ceil{\log_2(n)}$ where $n$ is the number of qubits. Then we present a specific strategy that
  gives the upper bound
  $\frac{4}{3}n +8\ceil{\log_2(n)}$. This is an
  improvement over the best algorithms in the literature for medium
  sized registers (between a few hundreds and several thousands qubits, see Table~\ref{summary_bounds} for more details).
\item We present greedy methods based on cost minimization techniques. They are complementary with DaCSynth in the sense that
they are best suited for small problem sizes or best-case scenario while DaCSynth is better for large problems or worst-case operators.
\item We propose an extension to the case where encoded ancillary qubits are used.
\item We also give benchmarks of our method to support our theoretical
  results and compare them to state-of-the-art algorithms. In a worst case, DaCSynth provides circuits of depth smaller than $n$ where $n$ is the
  number of qubits. This improves the state-of-the-art algorithms by a factor of 2. For small or best-case operators, the greedy methods provide almost optimal results.
\item Finally, we apply our algorithms to the optimization of a class of reversible functions, with and without ancillary qubits. Starting from a circuit with optimized T-depth, we re-synthesize every chunk of purely CNOT circuits. We manage to consistently reduce the total depth of the circuits while keeping the T-count and the T-depth unchanged. Overall, we reduce the depth in average by $47$\% ($58$\% with ancillary qubits) and up to $92$\% ($99$\% with ancillary qubits).
\end{itemize}

The plan of this paper is the following: in Section \ref{sec::background} we present some background about the synthesis of linear reversible circuits. In Section \ref{sec::algo} we describe a new divide and conquer algorithm and give some strict upper bounds on the depth of the circuits synthesized by our method. In Section \ref{sec::greedy} we describe the greedy algorithms based on cost minimization techniques. We take into account encoded ancillary qubits in Section~\ref{sec::extension}. Benchmarks are given in Section~\ref{sec::bench}. We discuss some future work in Section~\ref{sec::discussion} and we conclude in Section~\ref{sec::conclusion}.

\begin{table*}[tb]
  \centering
\begin{tabular}{c|cccccccc} \toprule
    Method & Gaussian  & \cite{DBLP:journals/cjtcs/KutinMS07} & \cite{DBLP:conf/soda/JiangSTW0Z20} & Our algorithm \\ 
    & elimination & & & \\ \midrule
    Upper bound   & $4n$ & $2n$ & $\mathcal{O}\left(\frac{n}{\log_2(n)}\right)$  & $ \frac{4}{3}n+8\log_2(n)$ \\
    Best result for $n$ such that & - & $ n<75$ & $  n > 1,345,000$ & $  75<n<1,345,000$ \\
\end{tabular}
\caption{Synthesis algorithms and theoretical upper bounds with the approximate ranges of validity for each method.}
\label{summary_bounds}
\end{table*}

\section{Background and state of the art} \label{sec::background}
\hfill \vspace*{-0.5cm}
\subsection{Notion of linear reversible function}
Let $\mathbb{F}_2$ be the Galois field of two elements. A
Boolean function $f : \mathbb{F}_2^n \to \mathbb{F}_2$ is
said to be linear if
\[
  f(x_1 \oplus x_2) = f(x_1) \oplus f(x_2)
\]
for any $x_1, x_2 \in \mathbb{F}_2^n$ where $\oplus$ is the bitwise
XOR operation. Let $e_k$ be the $k$-th canonical vector of
$\mathbb{F}_2^n$. By linearity we can write for any
$x = \sum_k \alpha_k e_k$ (with $\alpha_k \in \{0,1\}$)
\[
  f(x) = f\left(\sum_{k} \alpha_k e_k\right) = \sum_k \alpha_k
  f(e_k)
\]
and the function $f$ can be represented with a
column vector $\vec\alpha = [f(e_1), ..., f(e_n)]^T$ such
that
\(
  f(x) = \alpha \cdot x,
\)
where $\cdot$ stands for the scalar product on $F_2^n$ and
  $(-)^T$ is the matrix-transpose operation. This easily 
  extends to the $n$-input $m$-outputs functions $f : \mathbb{F}_2^n \to \mathbb{F}_2^m$ 
  where $f$ is defined by an $m \times n$ Boolean matrix $A$ such that 
\( f(x) = Ax. \)

In the case of reversible Boolean functions, $n=m$ and we have a
one-to-one correspondence between the inputs and the outputs. We
then consider $n$-inputs $n$-outputs functions $f$ for which the
equation
\( y = f(x) = Ax \) must have a unique solution for any
$y \in \mathbb{F}_2^n$. In other words the matrix $A$ must be
invertible in $\mathbb{F}_2$ and there is a one-to-one correspondence
between the linear reversible functions
of arity $n$ and the invertible Boolean matrices of size
$n$. This can be used to count the number of different linear
reversible functions of $n$ inputs (see, e.g., \cite{patel2008optimal}).
The application of two successive operators $A$ and $B$ is equivalent to the application of the operator product $BA$.

\subsection{Synthesis of linear reversible Boolean\texorpdfstring{\\*}{} functions}

We are interested in synthesizing general linear reversible Boo\-lean
functions into a reversible circuit, i.e., a series of elementary
reversible gates that can be executed on a suitable hardware. For
instance in quantum computing the Controlled-Not gate (CNOT) is used
in universal gate sets for superconducting and photonic qubits and
performs the following 2-qubit operation:
\[
  \text{CNOT}(x_1, x_2) = (x_1, x_1 \oplus x_2).
\]
Clearly the CNOT gate is a linear reversible gate. It can be shown to be
universal for linear reversible circuit synthesis: any
linear reversible function of arity at least 2 can be implemented by a reversible
circuit containing only CNOT gates. In this
paper we aim at producing CNOT-based reversible circuits for any
linear reversible functions.

In terms of matrices, a CNOT gate controlled by the line $j$ acting on
line $i \neq j$ can be written $E_{ij} = I + e_{ij}$ where $I$ is the
identity matrix and $e_{ij}$ the elementary matrix with all entries
equal 0 but the component $(i,j)$ of value $1$.  

Generally the synthesis of an operator is done by reducing it to the identity operator.
In our case we want to compute a sequence of 
$N$ elementary matrices such that 
\[ \prod_{k=1}^N E_{i_k,j_k}A = I. \]
Finally, using the fact that $E_{ij}^{-1} = E_{ij}$, we get
\[ A = \prod_{k=N}^1 E_{i_k,j_k}\]
and a circuit implementing $A$ is given by concatenating 
the CNOT gates with control $j_k$ and target $i_k$. 

This can be generalized to the case where we allow both left and right multiplication by elementary matrices and the possibility to permute the rows and columns of $A$ before and after the reduction to the identity operator. In other words, we look for two sequences of elementary matrices (of size $N_1$ and $N_2$) and three permutation matrices $P, P_1, P_2$ such that 
\[ \prod_{k=1}^{N_1} E_{i_k,j_k}P_1AP_2\prod_{k=1}^{N_2} E_{i_k,j_k} = P. \]
Even with such generalization, it is still possible to rearrange the product to write
\[ A = P' \times \prod_{k=1}^{N} E_{i_k,j_k}\]
where $N=N_1+N_2$ and $P'$ is a permutation matrix. We read this as a CNOT circuit followed by a qubit permutation. 

Left-multiplying the operator $A$ by $E_{ij}$ performs an elementary row operation:
\[
  r_i \leftarrow r_i \oplus r_j,
\]
writing $r_k$ for the $k$-th row of
$A$. Similarly, right-multiplying the operator $A$ by $E_{ij}$ performs an elementary column operation:
\[
  c_j \leftarrow c_i \oplus c_j,
\]
writing $c_k$ for the $k$-th row of
$A$.

Thus, synthesizing a linear reversible function into a CNOT-based
reversible circuit is equivalent to transforming an invertible
Boolean matrix $A$ to the identity by applying elementary row and column operations. 
For the rest of the paper we will consequently
privilege this more abstract point of view because it gives more
freedom and often appears clearer for the design of algorithms. We
note by $\text{Row}(i,j)$ the elementary row operation
$r_j \leftarrow r_i \oplus r_j$ and $\text{Col}(i,j)$ the elementary column operation
$c_j \leftarrow c_i \oplus c_j$.

In order to evaluate the quality of a synthesis of a linear reversible
circuit a couple of metrics can be considered. The size of the circuit
given by its number of CNOT gates is a first one: this gives the
total number of instructions the hardware has to perform to execute
the circuit. Due to the presence of noise when executing every logical
gate, it is of interest to have the shortest circuit possible. In this paper
we focus on the second metric which is the depth of the circuit, i.e., the number of time steps
the hardware needs to execute the circuit if we suppose that
non-overlapping gates are executed simultaneously. The depth is
closely related to the execution time of the circuit. In quantum
computing the time available to perform computations is limited due to the short decoherence time of the qubits. Therefore it is crucial to 
be able to produce shallow circuits for complex algorithms. 

\subsection{State of the art}

In this paper we focus on improving the depth of linear reversible circuits with a full qubit connectivity, meaning that any CNOT gate between any pair of qubits can be done --- equivalently this means that any row operation is available. Recently an algorithm that produces asymptotically optimal circuits in $O(n/\log_2(n))$ was proposed\cite{DBLP:conf/soda/JiangSTW0Z20}. The theoretical depth is given approximately by the formula 
\[ d = \alpha \frac{n}{\log_2(n)} + \beta \sqrt{n}\log_2(n).\]
A detailed description of the algorithm is given in Appendix~\ref{appendix::jiang} where we estimate $\alpha$ and $\beta$ both to $20$. Thus for practical values of $n$ this algorithm does not provide shallow circuits. 

To our knowledge, for practical register sizes, four algorithms were designed. Three of them provide similar results: the standard Gaussian elimination algorithm, the skeleton circuits in \cite{maslov2007linear} and a practical algorithm proposed in \cite{DBLP:conf/soda/JiangSTW0Z20}. All give circuits for which the depth is upper bounded by $4n$. In \cite{DBLP:journals/cjtcs/KutinMS07}, Kutin \textit{et al.} proposed an algorithm for computing linear reversible circuits for the Linear Nearest Neighbor architecture (LNN) in a depth of at most $5n$. In Appendix~\ref{appendix::kutin} we describe this algorithm and we show that it can actually be extended straightforwardly to an algorithm for a full qubit connectivity and the depth of the output circuits is upper bounded by $2n$. To our knowledge this algorithm then is the best algorithm when the number of qubits does not exceed a few thousands.

\subsection{Our contributions}

We exploit the promising idea developed in
\cite{DBLP:conf/soda/JiangSTW0Z20} to use a divide-and-conquer method
in order to produce shallow circuits for reasonable sizes of registers. First we show that synthesizing an operator via a divide-and-conquer algorithm is equivalent to zeroing binary matrices with a given set of elementary operations. This provides a general framework, DaCSynth, giving another view of the problem from which new strategies can be applied. 
The algorithm in \cite{DBLP:conf/soda/JiangSTW0Z20} can be regarded as
one particular strategy for this framework. Although not the
goal of this paper, this means that we can
recover the asymptotic optimal behavior of their algorithm.

We propose two strategies to solve this new problem: the first
one --- essentially theoretical --- is a block algorithm and gives improved upper bounds on the total
depth in the worst case. The second algorithm is a greedy one and aims
at producing the shallowest circuits possible such that they can be executed on a quantum
hardware in a near future. Overall, our first algorithm produces circuits whose depth is bounded by $ \frac{4}{3}n+8\ceil{\log_2(n)}$ where $n$ is the number of qubits, improving the result in \cite{DBLP:journals/cjtcs/KutinMS07} and \cite{DBLP:conf/soda/JiangSTW0Z20} for intermediate sized problems. A summary of the theoretical results of the different algorithms is given in Table~\ref{summary_bounds}. The benchmarks show that our second algorithm improves the actual depth by a factor of 2 compared to the extension of Kutin \textit{et al.}'s algorithm and synthesizes circuits of depth $n$ in the worst case.

We also study the use of purely greedy algorithms. The global idea is to use cost minimization techniques with different cost functions to find "quickly" a shallow circuit. Greedy methods generally give good results for small problem sizes or for simple operators, but at the cost of no theoretical guarantee. In our benchmarks we will observe similar characteristics: greedy methods are the best up to a certain point where their performance degrades because they do not exploit the specific structures of the problem.

Next, we extend our framework in the case where ancillary qubits are encoded outputs of the function, i.e., we want to synthesize an operator $A_{\text{out}} \in F_2^{m \times n}$ with $m > n$ with an input operator $A_{\text{in}} \in F_2^{m \times n}$. 
We propose a simple block algorithm and we show that the depth increases logarithmically with the number of ancillas.

Finally we integrate our algorithms ---DaCSynth and the greedy ones--- into the quantum compiler Tpar \cite{DBLP:journals/tcad/AmyMM14} and test our method on a set of well known reversible functions.
This gives an overview of the total depth of quantum circuits implementing important arithmetic functions like adders, multipliers, etc.

\section{The algorithm DaCSynth} \label{sec::algo}

Given an operator $A \in F_2^{n \times n}$ to synthesize, our proposed algorithm DaCSynth is a divide and conquer algorithm and consists in the following steps : 
\begin{enumerate}
	\item First compute a permutation matrix $P$ such that $PA
          = \begin{pmatrix} A_1 & A_2 \\ A_3 & A_4 \end{pmatrix}$ and
          $A_1 \in F_2^{\ceil{n/2} \times \ceil{n/2}}$ is
          invertible,
	\item Apply row operations on $A$ to zero the block $A_3$ such that the resulting matrix is $A' = \begin{pmatrix} A'_1 & A'_2 \\ 0 & A'_4 \end{pmatrix}$, 
	\item Apply row operations on $A'$ to zero the block $A'_2$ such that the resulting matrix is ${A''} = \begin{pmatrix} {A''}_1 & 0 \\ 0 & {A''}_4 \end{pmatrix}$, 
	\item Call recursively the algorithm on ${A''}_1$ and ${A''}_4$. When $n=1$ return an empty set of row operations.
\end{enumerate}

Step 1 is straightforward: consider the rows of the submatrix $A[:,1:\ceil{n/2}]$ (using Matlab notation). Start from an empty set and at each step add a row to the set. If the rank of the set is increased, keep the row otherwise remove it. If the resulting set is not of rank $\ceil{n/2}$ this would mean that the first $\ceil{n/2}$ columns of $A$ are not linearly independent which is impossible by invertibility of $A$. In addition, we assume that the qubits are fully connected so we can avoid to apply $P$ by doing a post processing on the circuit that would transfer the permutation operation directly at the end of the total circuit. This can be done without any overhead in the number of gates. Hence the core of the algorithm lies in steps 2 and 3. We now give the details for processing step 2. This can be easily transposed to do step 3 as well. 

\begin{theorem}
Given $A = \begin{pmatrix} A_1 & A_2 \\ A_3 & A_4 \end{pmatrix} \in F_2^{n \times n}$ with $A_1 \in F_2^{\ceil{n/2} \times \ceil{n/2}}$ invertible, zeroing $A_3$ by applying row operations on $A$ is equivalent to zeroing the matrix $B = A_3A_1^{-1}$ by applying any row and column operations on $B$ or flipping any entry of $B$.
\end{theorem}

\begin{proof}
First, note that by hypothesis $A_1$ is invertible so the matrix $B$ does indeed exist. 
\begin{itemize}
\item Applying an elementary row operation $\text{Row}(i,j)$ on $A_3$ gives the matrix $E_{ji}A_3$ and $B$ is updated by $E_{ji}B$. Thus a row operation on $A_3$ is equivalent to a row operation on $B$. 
\item Applying an elementary row operation $\text{Row}(i,j)$ on $A_1$ gives the matrix $E_{ji}A_1$ and $B$ is updated by $BE_{ji}$. Thus a row operation on $A_1$ is equivalent to a column operation on $B$.
\item $B$ is a $\floor{n/2} \times \ceil{n/2}$ matrix. The k-th row of $B$ gives the decomposition of the k-th row of $A_3$ in the basis given by the rows of $A_1$. Thus any row operation $\text{Row}(k_1, \ceil{n/2} + k_2)$ on $A$ will flip the entry $(k_2,k_1)$ of $B$. 
\end{itemize}
With these three types of operations available on $B$, the invertibility of $A_1$ is preserved. Thus when $B$ is zero necessarily $A_3$ is also zero. 
\end{proof}

Obviously flipping all the $1$-entries of $B$ is enough to reduce $A_3$
to the null matrix, but we are concerned with the shallowest way of
doing this. In the following we show how to compute the optimal depth
of the circuit zeroing $B$ using only the flipping operation.

\begin{theorem}
With the same notations, let $k$ be the maximum number of $1$-entries in one row or one column of $B$. Then if we use only the flipping operation we need a circuit of depth $k$ to zero $B$.
\label{theorem_matching}
\end{theorem}

\begin{proof}
We exploit a theoretical result about bipartite graph already used in \cite{DBLP:conf/soda/JiangSTW0Z20}. Consider the bipartite graph $G = (V_1, V_2, E)$ where each vertex of $V_1$ is a row of $A_1$, each vertex of $V_2$ is a row of $A_3$ and $B$ is the adjacency matrix of $G$. Any matching in $G$ represents a series of row operations that can be executed in parallel and that will zero some entries in $B$. If there is at most $k$ non-zero entries in each row and column of $B$ this means that the degree of $G$ is $k$ as well. Any bipartite graph of degree $d$ can be decomposed into exactly $d$ matchings \cite{kapoor2000edge}. Hence a circuit of depth $k$ is needed to transform $B$ into the null matrix. 
\end{proof}

We are now able to give a strict upper bound on the worst case result of the algorithm DaCSynth. 

\begin{corollary}
The depth of the circuits given by the algorithm DaCSynth is upper bounded by $2n + 2\ceil{\log_2(n)}$ with $n$ the number of qubits.
\end{corollary}

\begin{proof}
A first straightforward formula for the depth of the circuit output by the algorithm DaCSynth is 
\[ d(n) = d(\ceil{n/2}) + 2 \times d^*(\ceil{n/2})\]

where $d^*$ is the depth of the circuits computing parts 2 and 3. Using the result of the previous theorem we have 
\[ d^*(n) \leq n \]
So overall the depth of our circuit is upper bounded by 
\[ d(n) \leq d(\ceil{n/2}) + 2\ceil{n/2} \]

As $d(1) = 0$ and by exploiting the fact that $\ceil{\ceil{n/2}/2} = \ceil{n/4}$ we have
\[ d(n) \leq 2 \times \left(\sum_{k=1}^{\ceil{\log_2(n)}} \ceil{n/2^k} \right) \leq 2 \times \left(\sum_{k=1}^{\ceil{\log_2(n)}} n/2^k + 1 \right)\]
After simplification we have 
\[ d(n) \leq 2n + 2\ceil{\log_2(n)} \qedhere\]
\end{proof}

\subsection{\textbf{A block algorithm for steps 2 and 3}}

In order to improve the upper bound of our framework, we propose a block method for performing steps 2 and 3. Given an $n \times n$ matrix $B$ to zero and an integer $k < n$ such that $n = bk + r$, we divide $B$ into a matrix of $\ceil{\frac{n}{k}} \times \ceil{\frac{n}{k}}$ blocks: 
\begin{itemize}
  \item $\floor{\frac{n}{k}}^2$ are of size $k$,
  \item $\floor{\frac{n}{k}}$ are of size $k \times r$,
  \item $\floor{\frac{n}{k}}$ are of size $r\times k$ 
  \item and the lower right one is of size $r \times r$. 
\end{itemize}

If $B$ is of size $n \times (n+1)$ or $(n+1) \times n$ (which can happen if $A$ is of odd size) then some blocks on the edge will be of size $k \times (r+1)$ or $(r+1) \times k$. In any case as $r < k$ then $r+1 \leq k$ and the critical point is that all of these rectangular blocks are smaller than the $k \times k$ blocks. 

Now we consider each nonzero block as a $1$-entry in a $\ceil{\frac{n}{k}} \times \ceil{\frac{n}{k}}$ binary matrix that can be mapped to a bipartite graph $G$ as above. Then it is clear that a matching in $G$ corresponds to a subset of blocks on which we can apply row and column operations in parallel. 

Considering one such matching, we assume that we can reduce the maximum number of $1$-entries in each row and column of one block to an integer $p$ in depth at most $D$. Then all the blocks are matrices with at most $p$ nonzero entries per row and column and we can flip all of these non zero entries in $p$ sequences of row operations as they belong to different rows and columns in $B$. After that all the blocks of the matching are zero and we can repeat the process with another matching without modifying the nullified blocks. $G$ can be decomposed into at most $\ceil{\frac{n}{k}}$ matchings, each of them requires a depth of at most $D+p$ to zero all the blocks so the total depth for performing step 2 (or step 3) is $(D+p) \times \ceil{\frac{n}{k}}$. Again using the formula $\ceil{\frac{\ceil{n/m}}{k}} = \ceil{\frac{n}{mk}}$ an upper bound for the total depth is given by 
\[ d(n) \leq 2(D+p)\times \left(\sum_{j=1}^{\ceil{\log_2(n)}} \ceil{n/(k2^j)} \right). \]

After calculation we get 
\begin{equation} 
d(n) \leq \frac{2(D+p)}{k} n + 2(D+p)\ceil{\log_2(n)}.
\label{depth_equation}
\end{equation}

Note that with $k = 1$ then $D=0, p=1$ and we recover the result of Theorem~\ref{theorem_matching}. We are now ready to prove our main result. 

\begin{corollary}
The depth of the circuits given by the algorithm DaCSynth is upper bounded by $ \frac{4}{3}n + 8\ceil{\log_2(n)}$ with $n$ the number of qubits.
\end{corollary}

\begin{proof}
To improve our first result we need to find more efficient syntheses of our blocks. We performed a brute-force search for square matrices of size $k=1,2,3,4,5,6$. The search consisted in a breadth-first search: starting from the partial permutations, row/columns operations were applied in a growing depth manner. We explore the set of binary matrices and compute the minimum depth required to reduce them to a partial permutation. In order to reduce the size of the search we only considered matrices up to row and column permutations. For this purpose we used standard techniques involving graph isomorphism to compute a canonical representative for each class \cite{DBLP:journals/amco/Freibert13}. The results are given in Table~\ref{results_depth}. We recall that row and column operations can be performed in parallel. It is clear that for smaller rectangular cases the worst case depth cannot be larger. So by considering blocks of size $6$ the depth in the worst case is $D = 3$ and $p=1$. Replacing in Eq.~\ref{depth_equation} gives the result.
\end{proof}

We want to insist on the fact that the current upper bound is to be improved. In fact any improvements in the zeroing of larger blocks can significantly improve the theoretical upper bound and its range validity given in Table~\ref{summary_bounds}. For instance computing the worst case depth for $k=7,8,9$, if possible, may lead to a better upper bound. Moreover, what happens if we stop the row and column operations once the maximum number of $1$-entries in each row and column is below an integer $p > 1$\,? If $D$ decreases faster than $p$ increases this would represent another improvement.

As we already mentioned, the
synthesis algorithm proposed in \cite{DBLP:conf/soda/JiangSTW0Z20} can in fact be
seen as a special case where $A_1 = I$ and their strategy is also a
block algorithm with blocks of size $\log_2(n)/2 \times \log_2(n)/2$ and $n/\log_2(n) \times \log_2(n)/2$. 
Yet, translated in our framework, they only use operations on columns and the flipping entries operation.

\begin{table*}
\centering
\begin{tabular}{c|cccccccc} \toprule
    \backslashbox{Depth\kern-2em}{\kern-2emNumber of qubits} & 1 & 2 & 3 & 4 & 5 & 6 \\ \midrule
    0   & 2 & 3 & 4  & 5    & 6     & 7    \\
    1   &   & 4 & 17 & 69   & 199   & 630    \\
    2   &   &   & 15 & 243  & 5052  & 194390    \\
    3   &   &   &    &      & 367   & 56583   \\
\end{tabular}
\caption{Number of binary matrices reachable for different number of qubits and circuit depth (up to row/column permutations).}
\label{results_depth}
\end{table*}

\subsection{A greedy algorithm for steps 2 and 3} \label{sec::dac_greedy}

In practice we use a greedy algorithm to perform steps 2 and 3. We
recall that we work on a matrix $B$ that we want to zero with the
following three available operations: (1) row operations, (2) column
operations, (3) flipping one entry.
Note that row and column operations can be performed in parallel as this corresponds to CNOT circuits on two disjoint subsets of qubits.

At each step we compute a sequence of row and column operations on $B$ that
minimizes the number of ones in $B$ and that can be done in
parallel. If we only consider row or column operations then the optimal 
sequence can be computed in a polynomial time. To do so we create a directed graph $G_{\text{row}}$/$G_{\text{col}}$
whose nodes are the rows/columns of $B$ and the edges $(i \to j)$ are weighted by the gain in 
the number of ones if we apply the row operation $i \to j$. The optimal sequence of row/column operations
is given by the maximum weight matching in such graph which can be computed in polynomial time using the Blossom algorithm \cite{edmonds1965paths}.

\medskip
However, when considering both row and column operations, things are not that simple. A row operation on $B$ modifies 
$G_{\text{col}}$ and a column operation modifies $G_{\text{row}}$ so we cannot solve independently (or one after the other) the two problems in 
order to have an optimal sequence. 
The maximum weight matching problem on $G = (V,E)$ can be reformulated
as a linear programming problem

\begin{problem}
Maximize
\[
\sum_{e \in E}  x_ew(e) 
\]
such that for all vertices $u \in V$, 
\[ \sum_{e \in \{ (u,v), (v,u) | v \in \delta(u) \} } x_e \leq 1\]
and for all edges $e \in E$, $x_e \in \{0,1\}$  
\end{problem}

where $\delta(u)$ stands for the set of nodes adjacent to $u$.
Taking into account both row and column operations adds quadratic terms in the cost function and that complicates the search for an optimal solution. 
Namely this new problem on the two graphs $G_{\text{row}} =
(V_{\text{row}}, E_{\text{row}})$ and $G_{\text{col}} =
(V_{\text{col}}, E_{\text{col}})$ can be reformulated as 
\begin{problem}\label{np_hard_problem}
Maximize
\begin{multline*}
\sum_{e_{\text{row}} \in E_{\text{row}}}          x_{e_{\text{row}}} w(e_{\text{row}}) 
+ \sum_{e_{\text{col}} \in E_{\text{col}}}          x_{e_{\text{col}}}w(e_{\text{col}}) \\
+ \sum_{e_{\text{col}}, e_{\text{row}}}
x_{e_{\text{row}}}x_{e_{\text{col}}} q(e_{\text{row}}, e_{\text{col}})
\end{multline*}
such that for all vertices $u\in V_{\text{row}}$,
\[\sum_{e_{\text{row}} \in \{ (u,v), (v,u) | v \in \delta(u) \} }
  x_{e_{\text{row}}} \leq 1,\]
for all vertices $u\in V_{\text{col}}$,
\[
\sum_{e_{\text{col}} \in \{ (u,v), (v,u) | v \in \delta(u) \} }
x_{e_{\text{col}}} \leq 1
\]
and for all edges $e \in E_{\text{row}} \cup E_{\text{col}}$, $x_e \in
\{0,1\}$.
\end{problem}

where the $q$'s are the quadratic terms. Each quadratic term corresponds to a specific entry in $B$ so we have $q(e_{\text{row}}, e_{\text{col}}) \in \{-1, 0, 1\}$. Problem~\ref{np_hard_problem} is a particular instance of the \emph{quadratic matching} problem where, given a graph $G$, one must find a matching that optimizes an objective function containing linear terms on the edges and quadratic terms on the pairs of edges. In Problem~\ref{np_hard_problem}, the graph $G$ is given by the disjoint union of the two graphs $G_{\text{row}}$ and $G_{\text{col}}$, and the quadratic terms between two edges of $G_{\text{row}}$ or two edges of $G_{\text{col}}$ are 0. The quadratic matching problem is known to be NP-hard \cite{klein2014}, is Problem~\ref{np_hard_problem} also NP-hard? We leave this question as a future work.

We still tried to solve exactly Problem~\ref{np_hard_problem} with an integer programming solver. Yet, given the quadratic terms the number of variables and constraints evolves as $n^4$ where $n$ is the number of qubits and the method cannot find a solution even for $n=10$.
To get a non optimal solution in a reasonable amount of time, we compute a sequence of row and column operations greedily.
We first choose the best row or column operation that minimizes the number of
ones in $B$ and we keep in memory the operation applied. Then we determine the next best row or column operation among the operations that can be performed in parallel with the previously stored operation and we repeat the process. Finally if some rows and columns are left untouched we may complete the sequence of operations by flipping some entries. The best sequence of flipping operations is computed as described in the proof of Theorem~\ref{theorem_matching} using the Blossom algorithm. If no row or column operation can reduce the number of $1$ in $B$ then only the flipping operation is used.

\section{Purely greedy algorithms} \label{sec::greedy}

During steps 2 and 3 of the DaCSynth algorithm in Section~\ref{sec::dac_greedy}, we used a greedy process to zero a boolean matrix with as few operations as possible. We now explore the use of similar techniques directly on the linear boolean reversible operator to synthesize. It has been proven to be efficient for size optimization \cite{de2021gaussian}. The method consists in a cost minimization technique, we need:
\begin{itemize}
  \item a cost function to minimize,
  \item a strategy to explore the set of linear reversible operators.
\end{itemize}
Similarly to \cite{de2021gaussian}, we consider the following four cost functions to guide our search: 
\begin{itemize}
  \item $h_{\text{sum}}(A) = \sum_{i,j} A_{i,j},$
  \item $H_{\text{sum}}(A) = h_{\text{sum}}(A) + h_{\text{sum}}(A^{-1}),$
  \item $h_{\text{prod}}(A) = \sum_{i} \log_2(\sum_j A_{i,j}),$
  \item $H_{\text{prod}}(A) = h_{\text{prod}}(A) + h_{\text{prod}}(A^{-1}).$
\end{itemize}

These four cost functions reach their minimum when $A$ is a permutation matrix, motivating their use in a cost minimization process. If the cost function $h_{\text sum}$ seems the first natural choice, the cost function $h_{\text prod}$ has interesting features because it gives priority to "almost done" rows. Namely, if one row has only a few nonzero entries, the minimization process with $h_{\text prod}$ will treat this row in priority and then it will not modify it anymore. This enables to avoid a problem which one meets with the cost function $h_{\text sum}$ where one ends up with a very sparse matrix but where the rows and columns have few nonzero common entries. This type of matrix represents a local minimum from which it can be difficult to escape. With this new cost function, as we put an additional priority on the rows with few remaining nonzero entries, we avoid this pitfall. Adding the cost of the inverse matrix also helps to escape from local minima.

In order to choose which row and column operations to apply, we proceed similarly to the DaCSynth algorithm: we keep track of previously applied row and column operations to determine which supplementary operations can be done without increasing the depth. At each iteration we choose among the remaining operations that actually decrease the cost function the one that minimizes the cost function. If there are several possible operations, a random one is chosen. If no row or column operations can decrease the cost function, we reset simultaneously the set of row and column operations available.
Every time the set of applied row or column operations is reset we increase a counter by 1. The algorithm stops whether the current operator is a permutation matrix or when the counter exceeds a certain threshold. 

We know from previous experiments \cite{de2021gaussian} that such purely greedy algorithms behave extremely well on small operators (typically $n < 40$) or operators that need small/shallow circuits to be implemented. After a certain operator size or ``complexity'', the cost minimization process falls into local minima from which it is impossible to escape without a prohibitive overhead in the number of CNOTs or in the depth. One proposal to mitigate this behavior is to rely on an LU decomposition. It is well-known that any operator $A$ can be written 
\( A = PLU\)
where P is a permutation matrix and $L, U$ are triangular operators. We considered the case where we use our greedy algorithm on those triangular operators (the concatenation of the circuits obtained give a circuit for $A$ up to the permutation $P$) with the hope that the expected bad scalability is mitigated at the price of worse results when the purely greedy methods perform well. Several triplets $(P, L, U)$ are possible for one operator $A$. It was shown in \cite{de2021gaussian} that it is possible to adopt specific strategies to compute $(P,L,U)$. One of them consists in choosing iteratively the columns of $L$ and rows of $U$ to be the sparsest possible, this will be our approach and we will refer to this strategy as "LU sparse" in the benchmarks.

\section{Extension with ancillary qubits} \label{sec::extension}

The quantum compiler Tpar efficiently reduces the T-depth by computing subsets of T gates that can be applied in parallel \cite{DBLP:journals/tcad/AmyMM14}. Each T gate is associated
to a parity, i.e, a linear combination of the input qubits. With a subset of parities that are linearly independent, they can be computed at the same time 
and the T gates are applied in parallel to each qubit carrying one of these parities.

With ancillary qubits the parallelization can be even more efficient because the ancillary qubits can carry any parity, i.e, it can be a linear combination of the parities carried by non ancillary qubits. In terms of CNOT circuits synthesis we need to synthesize a larger linear reversible operator. Namely we have to synthesize Boolean matrices of size $p \times n$ where $p-n$ is the number of additional qubits that will carry a parity. We extend our framework to treat this particular case. Our goal is, given an input operator $A_{\text{in}} \in F_2^{p \times n}$ and an output operator $A_{\text{out}} \in F_2^{p \times n}$, to synthesize an operator $B \in F_2^{p \times p}$ such that $BA_{\text{in}} = A_{\text{out}}$. The main difference with standard linear reversible circuit synthesis is that $B$ is not unique so we need to find a suitable $B$ and to synthesize it.

We propose a simple block algorithm and we prove that the total depth for the synthesis is equal, up to additive logarithmic terms, to the depth 
of the synthesis on an operator $A \in F_2^{p \times n}, n \leq p \leq 2n$. This result shows that the total depth barely increases with the number of ancillas after a certain threshold. 

\subsection*{A block extension algorithm}

Let $A \in F_2^{p \times n}$. We assume that the first $n$ rows of $A$ form an invertible matrix. If not, we can always find a permutation matrix $P$ such that $PA$ is as desired, see Section~\ref{sec::algo}. Given $p = kn + r$, we partition the operator $A$ into $k$ blocks of $n$ rows and one block of $r$ rows. 
As assumed the first block is of full rank and we merge it with the block of $r$ rows. 

The core of this extension algorithm lies in the idea that it is cheap to make each block invertible. Actually it can be done with a circuit of depth $\ceil{\log_2(k)}$ by using the following lemma: 
\begin{lemma}
Given two matrices $A, B \in F_2^{n\times n}$ with $A$ of full rank. There exists a partial permutation $P$ such that 
$B + PA$ is of full rank.
\end{lemma}

\begin{proof}
Suppose $B$ is of rank $k, k<n$. We can write $B = CD$ where $C \in F_2^{n \times k}, D \in F_2^{k \times n}$ are of rank $k$. 
One can always add a set of $n-k$ canonical vectors to the columns of $C$ 
to create a basis of $F_2^n$. We can complete as well the rows of $D$ into a basis of $F_2^n$ by adding row vectors of $A$. We
get two new extended matrices $C', D'$ such that $B' = C'D'$ is now invertible. We can always add zero columns in $C'$ and 
the remaining rows of $A$ in $D'$ and reorder the columns of $C'$ and $D'$ to get $C' = [C | P]$ and $D' = [D | A]$ with $P$ a partial
permutation matrix. Rewriting $B' = [C|D] \times [P | A] = CD + PA = B + PA$ proves the result.
\end{proof}

To make all blocks invertible, we first make sure that the second block is of full rank by adding the appropriate rows of the first block. Using the lemma above
this operation can be done with a circuit of depth $1$.
Then we make sure that the third and fourth block are of full rank by adding the appropriate rows of the first and second block etc. Repeating this procedure, it is clear
that we only need $\ceil{\log_2(k)}$ iterations to treat all the blocks. 

What we want is an algorithm that synthesizes a circuit outputting an operator $A_{\text{out}} \in F_2^{p \times n}$ given an input operator $A_{\text{in}} \in F_2^{p \times n}$. Our proposal is to synthesize two operators $B_1, B_2$ such that the block partition of 
\[ B_1 A_{\text{in}} = \begin{pmatrix} K_1 \\ K_2 \\ \vdots \\ K_k \end{pmatrix}
 \text{\qquad and \quad}
 B_2 A_{\text{out}} = \begin{pmatrix} H_1 \\ H_2 \\ \vdots \\ H_k \end{pmatrix} \]
contain only invertible blocks. Then we can apply independently a linear reversible operator on each block to do the transition $B_1 A_{\text{in}} \to B_2 A_{\text{out}}$. Namely, for any $i > 1$, we apply $D_i = H_iK_i^{-1}$ to the $i$-th block. For $i=1$ we need an operator $D_1$ to do the transition $K_1 \to H_1 \in F_2^{(n+r) \times (n+r)}$. We did not particularly optimize this part, but as we know that the first $n$ rows of $K_1$ form an invertible matrix, we consider an operator $D_1$ of the form 
\[ D_1 = \begin{pmatrix} H_1[1:n,:]K_1[1:n,:]^{-1} & 0 \\ G & I_{n-r}\end{pmatrix} \]
where each row of $G$ contains the decomposition of each vector $K_1[i,:] \oplus H_1[i,:], i > n$ in the basis $K_1[1:n,:]$.

Overall we apply a block diagonal operator $D = \oplus_{i=1}^k D_i$ with $D_1 \in F_2^{(n+r) \times (n+r)}$ and $D_i \in F_2^{n \times n}$ for $i > 1$ such that
\[ D B_1 A_{\text{in}} = B_2A_{\text{out}} \]
and finally 
\[A_{\text{out}} = B_2^{-1} D B_1 A_{\text{in}}. \]

The total depth of our circuit is the sum of the depth of the circuits implementing $B_1, B_2$ and $D$. We know that the depth for implementing $B_1, B_2$ does not exceed $\ceil{\log_2(\floor{p/n})}$ and most of the total depth lies in the synthesis of $D$. The synthesis of $D$ requires a call to our framework for the square blocks of size $n$ and one call for the first block of size $n+r$. 
All syntheses are performed simultaneously so the total depth for step 2 is given by the maximum depth required for the synthesis
of one of the blocks. 

The total depth $d(n,p)$ is given by 
\[ d(n,p) = 2\log_2(\floor{p/n}) + d^*(n+r) \leq 4n + 2\log_2(\floor{p/n})\]
where $d^*(n+r)$ is the depth required to synthesize the block of size $(n+r)\times (n+r)$, which 
represents informally the maximum depth required when synthesizing all the blocks in parallel.
The upper bound is not the tightest possible. The result we want to emphasize is that the depth 
only depends logarithmically on the number of ancillas. 

\section{Benchmarks} \label{sec::bench}
This section presents our experimental results. We have the following algorithms to benchmark: 
\begin{itemize}
  \item DaCSynth from Section~\ref{sec::algo}, 
  \item cost minimization techniques from Section~\ref{sec::greedy},
  \item the extension of DaCSynth for the use of ancillary qubits described in Section~\ref{sec::extension}.
\end{itemize}

The state-of-the-art algorithms are the following: 
\begin{itemize}
  \item the Gaussian elimination algorithm,
  \item the algorithm from \cite{DBLP:journals/cjtcs/KutinMS07} adapted to a full qubit connectivity, as described in Appendix~\ref{cnot::soa_depth}.
\end{itemize}

Two kinds of data-sets are used to benchmark our algorithms: 
\begin{itemize}
  \item First, a set of random operators. The test on random operators gives an overview of the average performance of the algorithms. We generate random operators by creating random CNOT circuits. Our routine takes two inputs: the number of qubit $n$ and the depth $d$ desired for the random circuit. Each CNOT is randomly placed by selecting a random control and a random target and the simulation of the circuit gives a random operator. Empirically we noticed that when $d$ is sufficiently large --- $d=2n$ is enough --- then the operators generated have strong probability to represent the worst case scenarii. Alternatively, when only worst-case operators are of interest, it is faster to generate random circuits with a sufficiently large number of CNOT gates ($n^2$ gates is enough) instead of creating a circuit with a large depth.
  \item Secondly, a set of reversible functions, given as circuits, taken from Matthew Amy's github repository \cite{meamy}. This experiment shows how our algorithms (DaCSynth and the greedy procedures) can optimize useful quantum algorithms in the literature like the Galois Field multipliers, integer addition, Hamming coding functions, the hidden weighted bit functions, etc.
\end{itemize}

To evaluate the performance of our algorithms for the random set, two types of experiments are conducted:
\begin{enumerate}
  \item a worst-case asymptotic experiment, namely for increasing problem sizes $n$ we generate circuits of depth $2n$ and we compute the average depth for each problem size. This experiment reveals the asymptotic behavior of the algorithms and gives insights about strict upper bounds on their performance.
  \item a close-to-optimal experiment, namely for one specific problem size we generate operators with different circuit depth to show how close to optimal our algorithms are if the optimal circuits are expected to be shallower than the worst case.
\end{enumerate}

To produce the benchmarks we need an explicit way to compute the depth. This task is less trivial than computing the number of gates in the circuits. The most common way to perform this computation is to create a Directed Acyclic Graph representation of the circuit: the vertices of the graph are the gates and the edges represent their inputs/outputs. The depth of the circuit is then given by the longest path in the graph which can be computed by doing a topological sorting of the vertices for example. Another way is to divide the circuit into slices of parallel gates. When a new gate is added to the circuit one has to pull it to its maximum to the left of the circuit by commuting it with the existing slices. If the gate cannot commute with the first slice it encounters, a new slice is created. The number of slices is then equal to the depth of the circuit.
An interesting feature of this procedure is that we recover the skeleton circuit outlined in \cite{maslov2007linear} by computing the depth on a circuit returned by a standard Gaussian elimination algorithm. 

\paragraph{}
All our algorithms are implemented in Julia \cite{bezanson2017julia} and executed
on the ATOS QLM (Quantum Learning Machine) whose processor is an Intel
Xeon(R) E7-8890 v4 at 2.4 GHz.

\subsection{Benchmarks on random circuits}

We did the following experiments : 
\begin{itemize}
	\item we evaluated the worst case performance of the different algorithms for a range of qubits. For $n=1...100$, we tested the algorithm on 20 random circuits with high depth $=2n$ to reach with high probability the worst cases.
	\item We also evaluated the capacity of the different algorithms to find shallow circuits for a specific problem size. For $n = 60$, we tested our algorithms on random circuits of various depth from $1$ to $ \approx 80$ with $20$ circuits for each depth.
\end{itemize}

\subsubsection{Evaluation of DaCSynth}

\begin{figure}[tbp]
\includegraphics[scale=0.42]{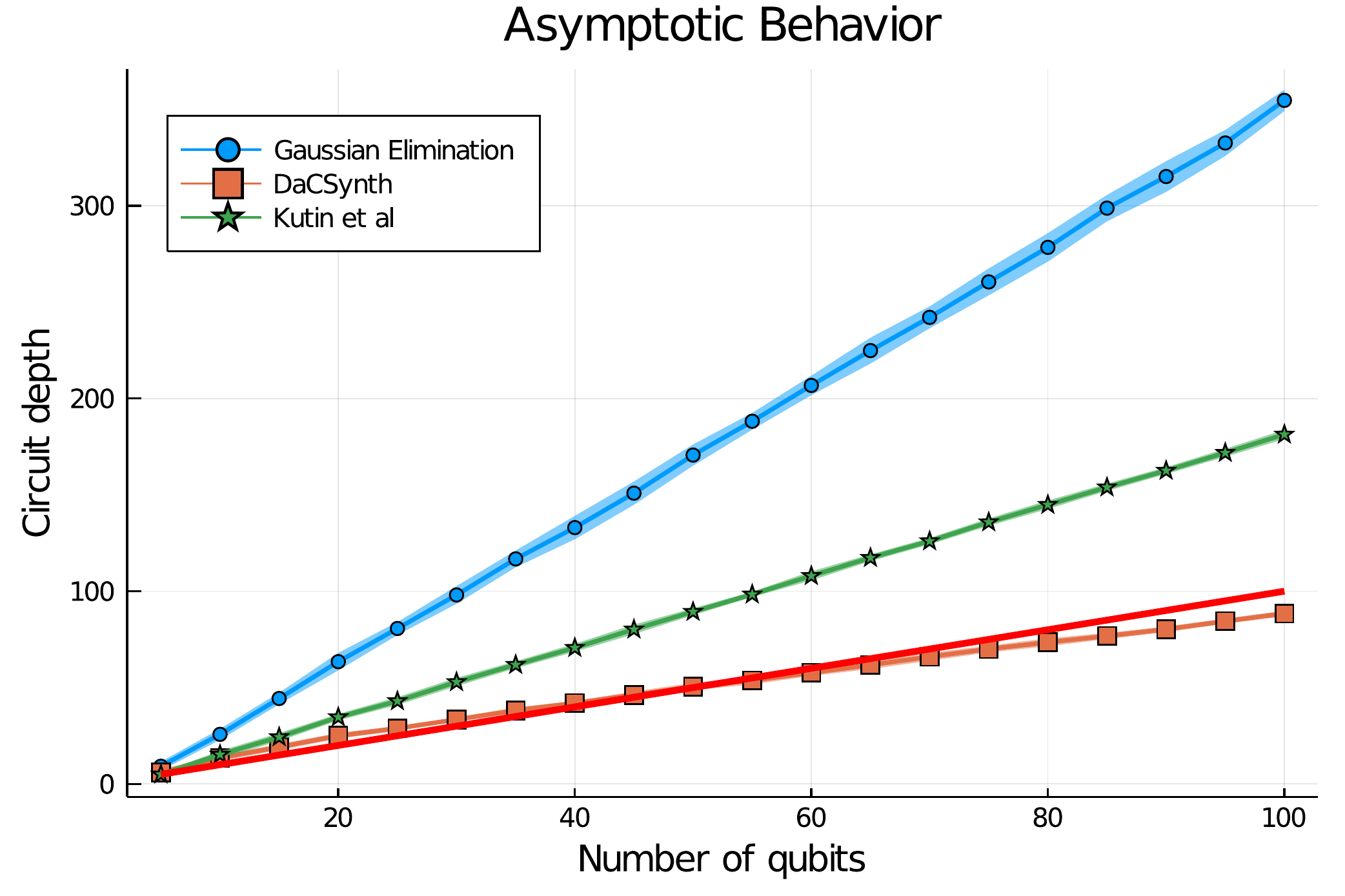}
\caption{Average performance of DaCSynth vs Gaussian elimination algorithm and \cite{DBLP:journals/cjtcs/KutinMS07}.}
\label{asymptotic_bench_dac}
\end{figure} 

For clarity we do not show all the methods at once. We first show the worst case performance of DaCSynth against the Gaussian elimination algorithm and Kutin \textit{et al.}'s algorithm in Figure~\ref{asymptotic_bench_dac}. In this case, both three algorithms have a linear complexity and we almost recover the theoretical worst case complexities: $\approx 4n$ for the Gaussian elimination algorithm, $\approx 2n$ for Kutin \textit{et al}'s algorithm. The depth complexity of DaCSynth is close to $n$ when $n < 50$ but tends to $0.85n$ when $n > 50$. For larger values of $n$ not shown in this graph ($100 < n < 1000$) the depth complexity seems to remain around $0.85n$ so we cannot really say if this complexity actually hides a complexity in $n/\log_2(n)$ or not. Our current implementation cannot deal with larger number of qubits, it would be interesting to implement a more efficient version of DaCSynth to see how DaCSynth behaves and also to do a proper comparison with the algorithm from \cite{DBLP:conf/soda/JiangSTW0Z20}.

DaCSynth outperforms the state of the art by at least a factor of $2$ and this outperformance is also visible in the close-to-optimal experiment given in Figure~\ref{60qubits_dac}. DaCSynth still is able to re-synthesize a circuit with small depth, although it cannot give optimal results. Overall, the behavior of those three algorithms (DaCSynth, the Gaussian Elimination and the extension of Kutin \textit{et al.}'s algorithm) is similar for any problem size. Consequently for the rest of the benchmarks we now consider DaCSynth as the state-of-the-art method.

\begin{figure}[tbp]
\includegraphics[scale=0.42]{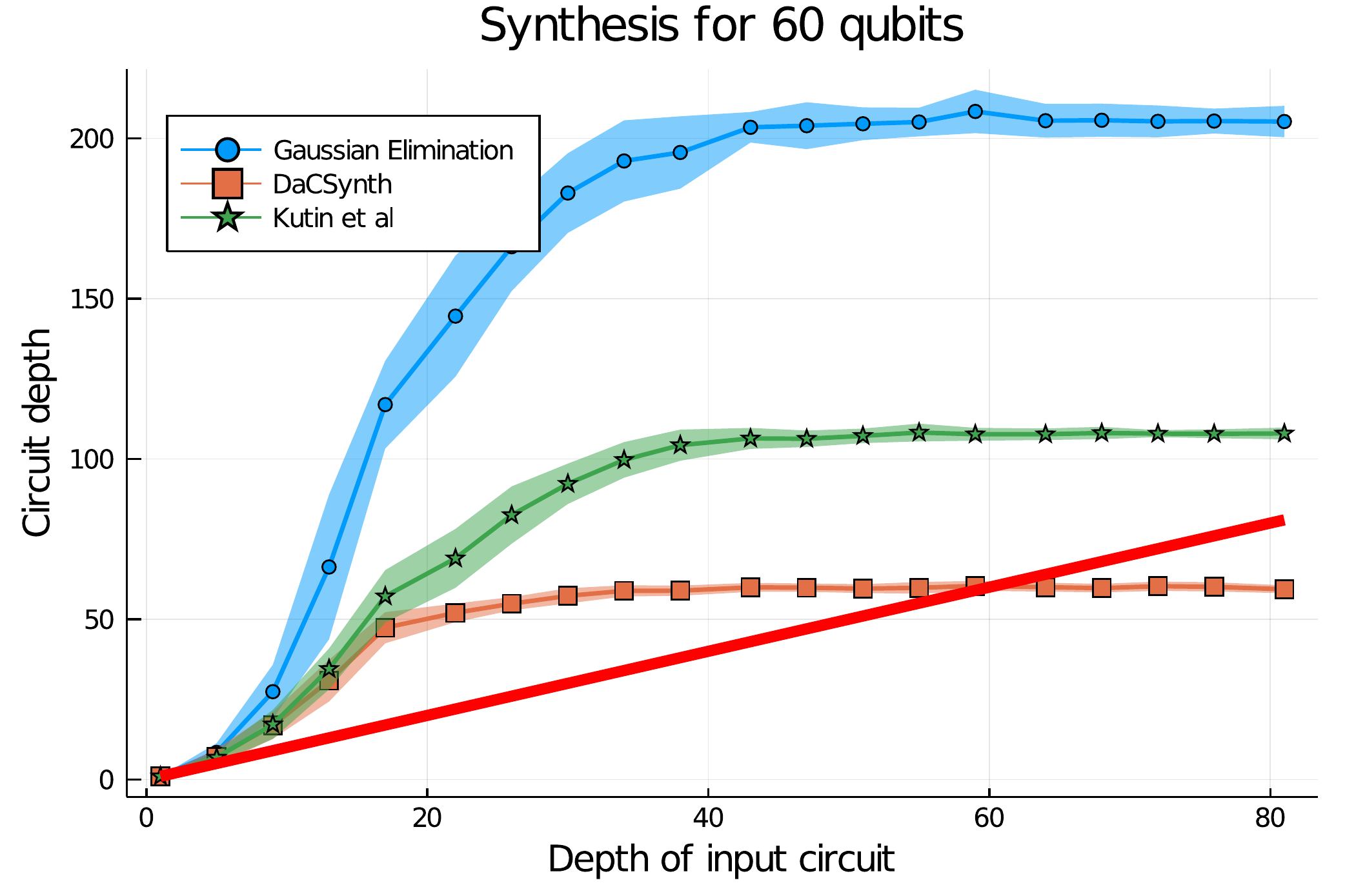}
\caption{Performance of DaCSynth vs Gaussian elimination algorithm and \cite{DBLP:journals/cjtcs/KutinMS07} on 60 qubits for different input circuits depths.}
\label{60qubits_dac}
\end{figure}
\begin{figure}[htbp]
\includegraphics[scale=0.42]{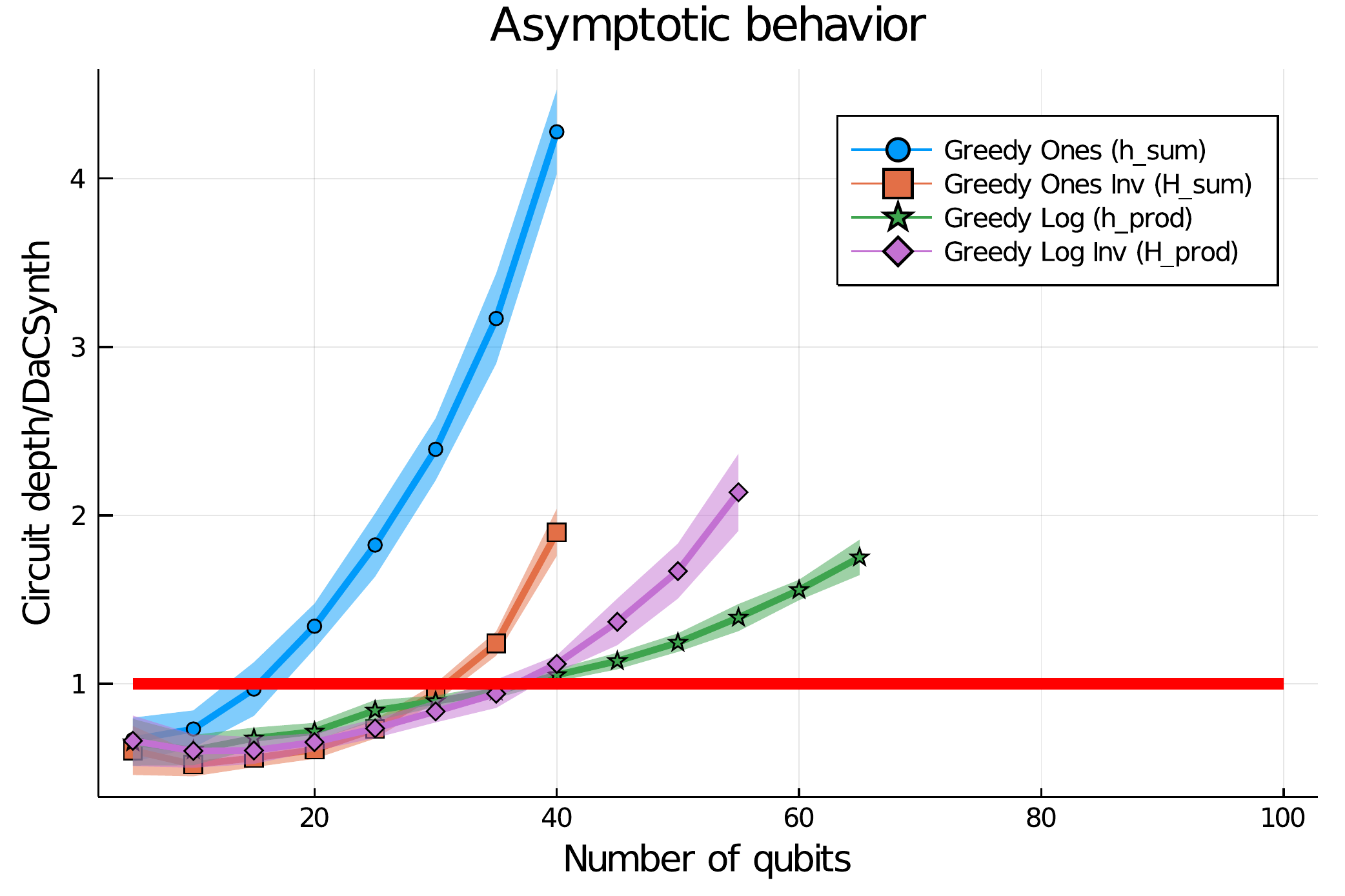}
\caption{Average performance of cost minimization techniques vs DaCSynth.}
\label{asymptotic_bench_greedy}
\end{figure} 

\begin{figure*}
\begin{subfigure}{0.48\textwidth}
  \includegraphics[scale=0.35]{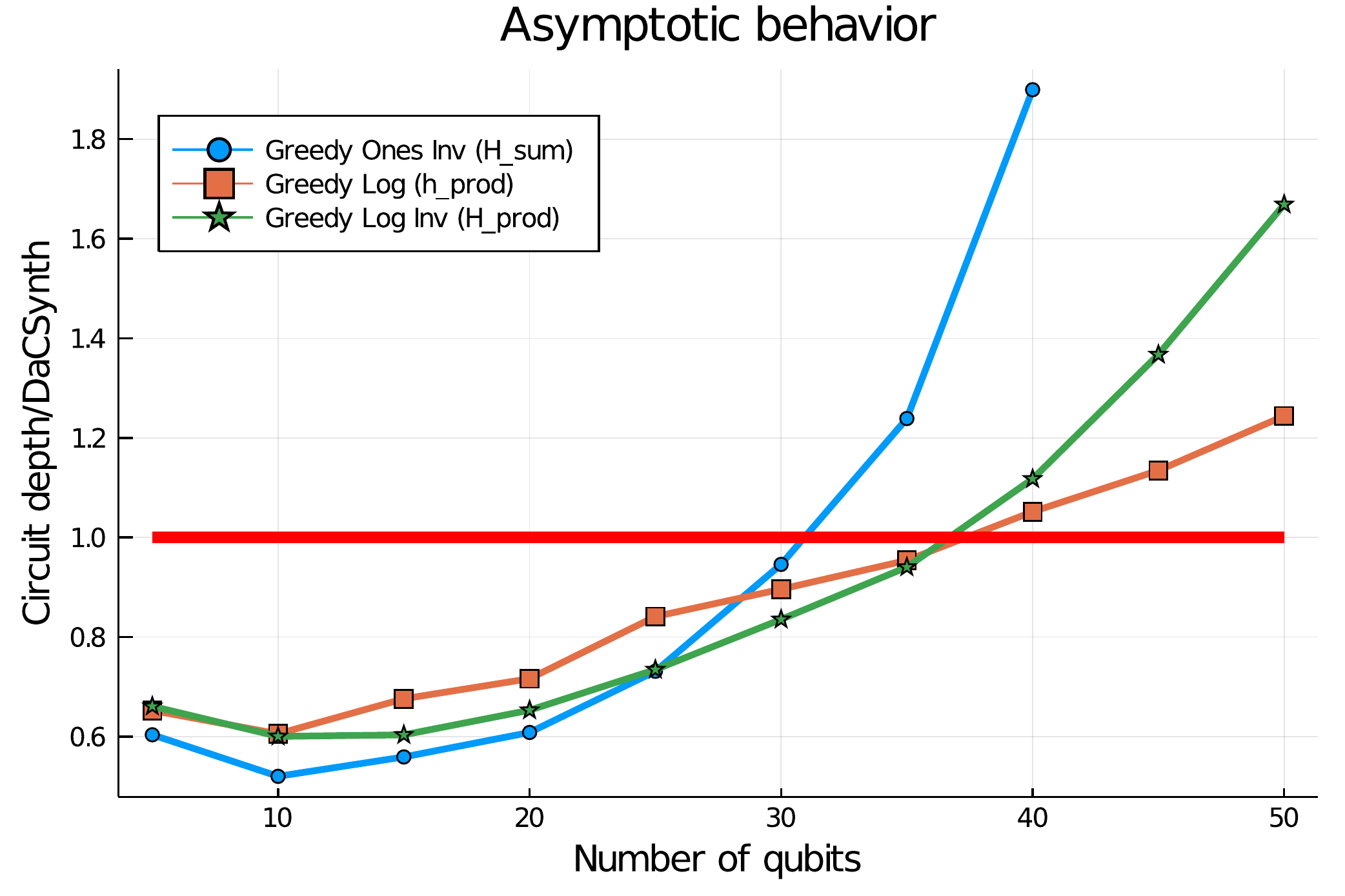}
  \caption{Average performance of cost minimization techniques (\{$h_{\text prod},H_{\text sum},H_{\text prod}$\}) vs DaCSynth (without standard deviation).\\ \hfill \\}
  \label{asymptotic_bench_greedy_zoom}
\end{subfigure}
\hfill
\begin{subfigure}{0.48\textwidth}
  \includegraphics[scale=0.35]{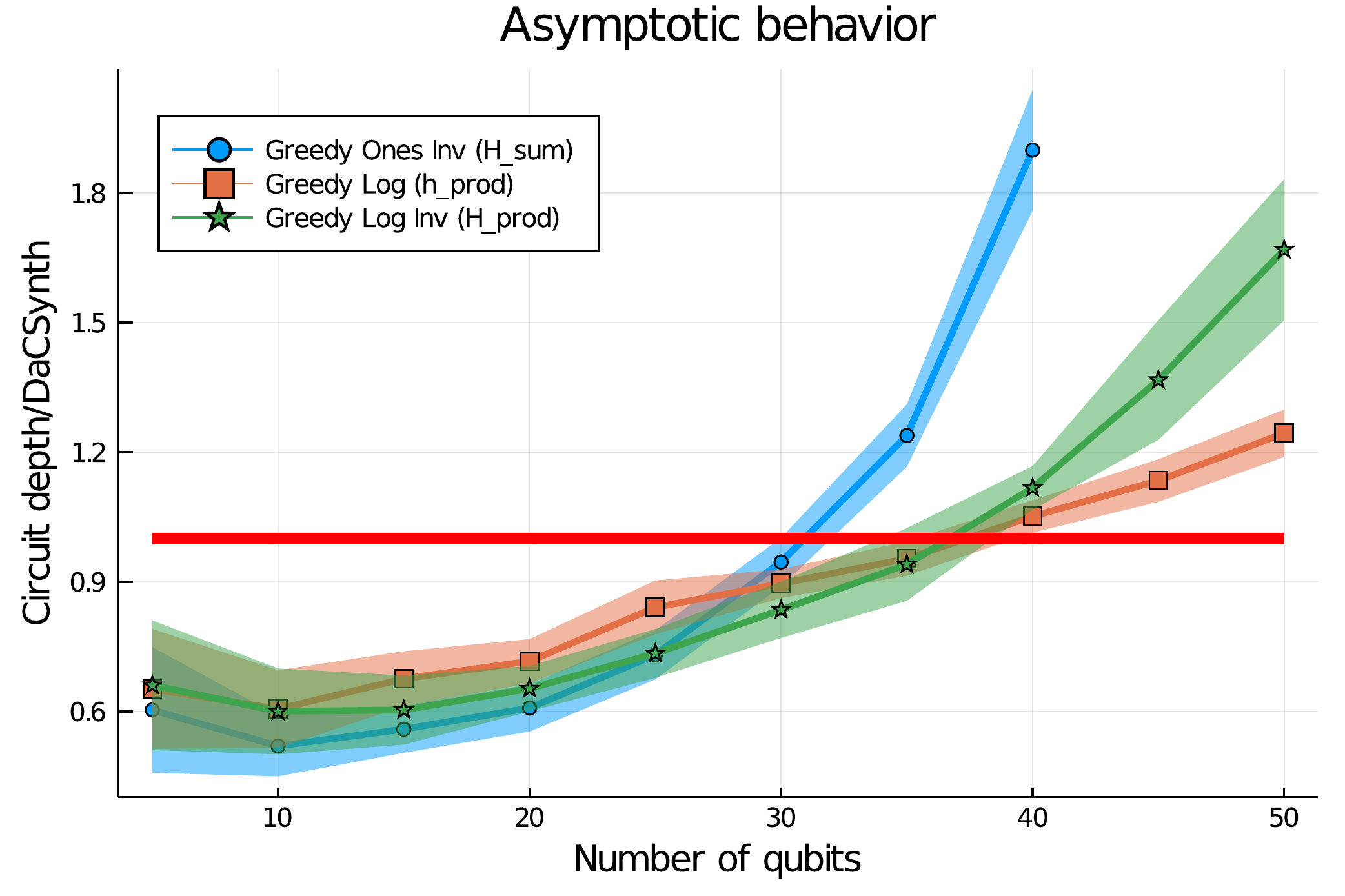}
\subcaption{Average performance of cost minimization techniques (\{$h_{\text prod},H_{\text sum},H_{\text prod}$\}) vs DaCSynth (with standard deviation).}
\label{asymptotic_bench_greedy_zoom_std}
\end{subfigure}
\caption{Average performance of cost minimization techniques (\{$h_{\text prod},H_{\text sum},H_{\text prod}$\}) vs DaCSynth.}
\label{depth_average_greedy}
\end{figure*}
\begin{figure*}
\begin{subfigure}{0.48\textwidth}
  \includegraphics[scale=0.40]{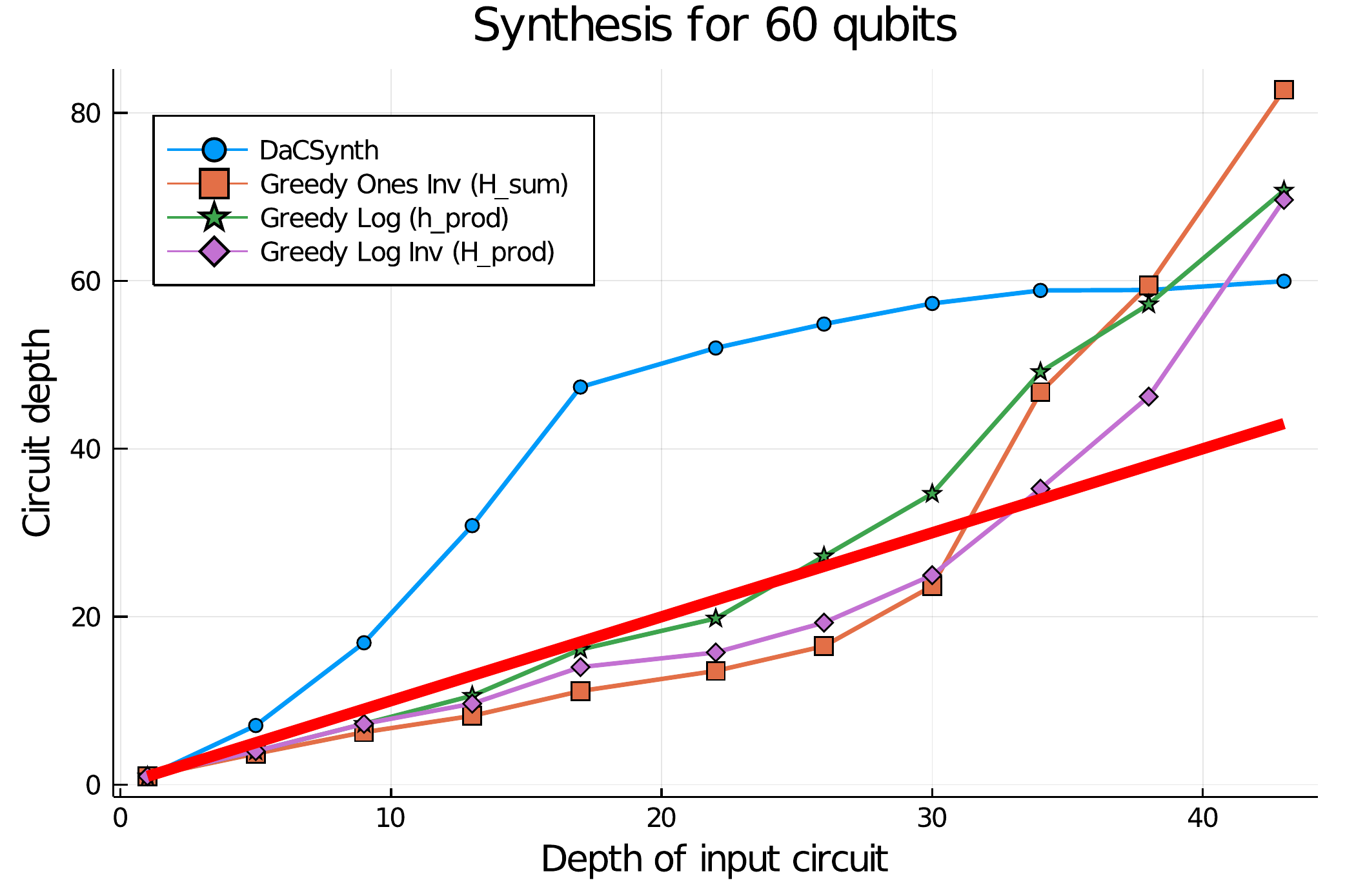}
  \caption{Performance of cost minimization techniques (\{$h_{\text prod},H_{\text sum},H_{\text prod}$\}) and DaCSynth on 60 qubits for different input circuits depths (without standard deviation).}
  \label{60qubits}
\end{subfigure}
\hfill
\begin{subfigure}{0.48\textwidth}
  \includegraphics[scale=0.40]{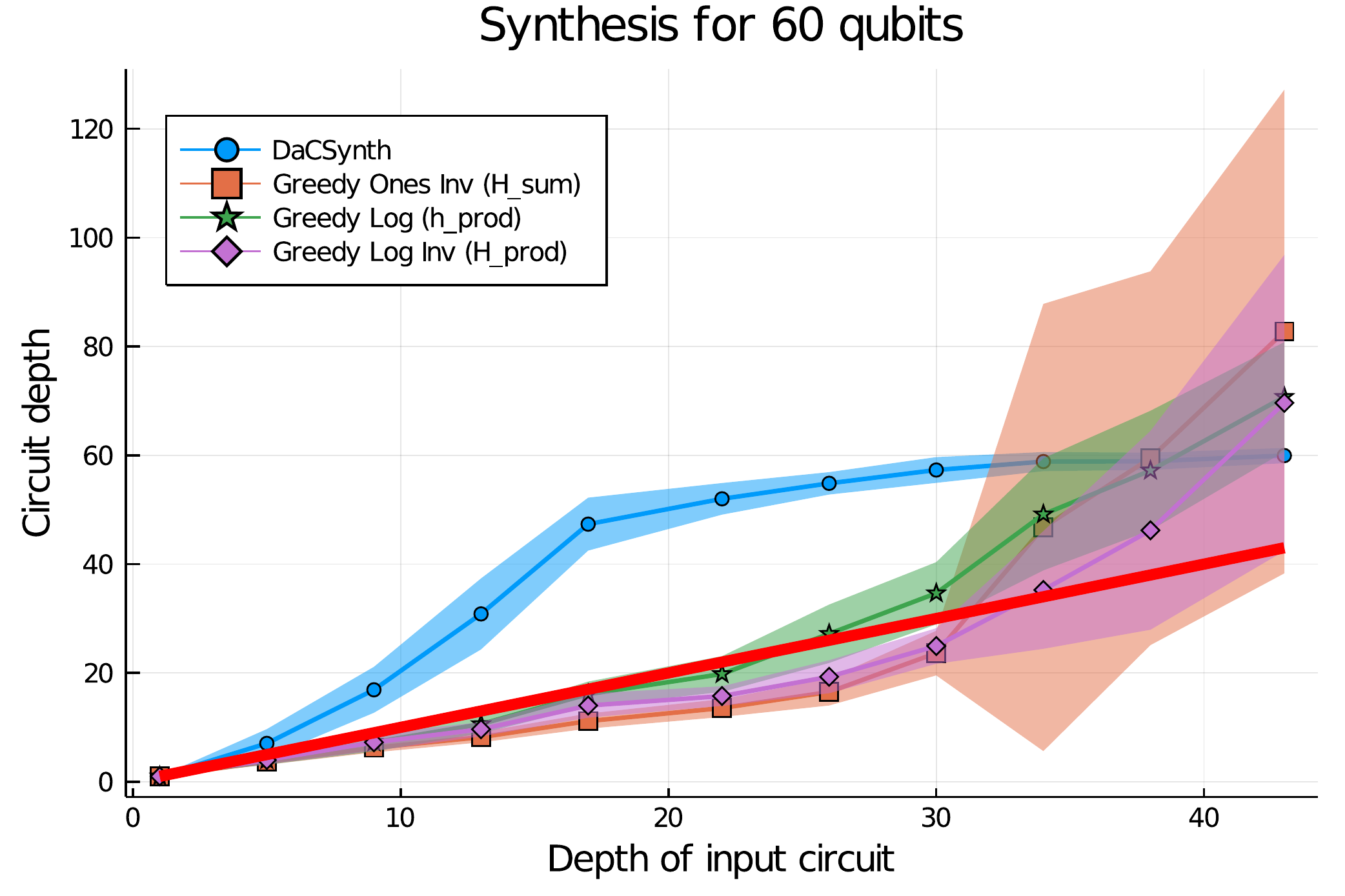}
\subcaption{Performance of cost minimization techniques (\{$h_{\text prod},H_{\text sum},H_{\text prod}$\}) vs DaCSynth on 60 qubits for different input circuits depths (with standard deviation).}
\label{60qubits_std}
\end{subfigure}
\caption{Performance of cost minimization techniques (\{$h_{\text prod},H_{\text sum},H_{\text prod}$\}) and DaCSynth on 60 qubits.}
\label{60qubits_perf}
\end{figure*}
\begin{figure*}
\begin{subfigure}{0.48\textwidth}
  \includegraphics[scale=0.40]{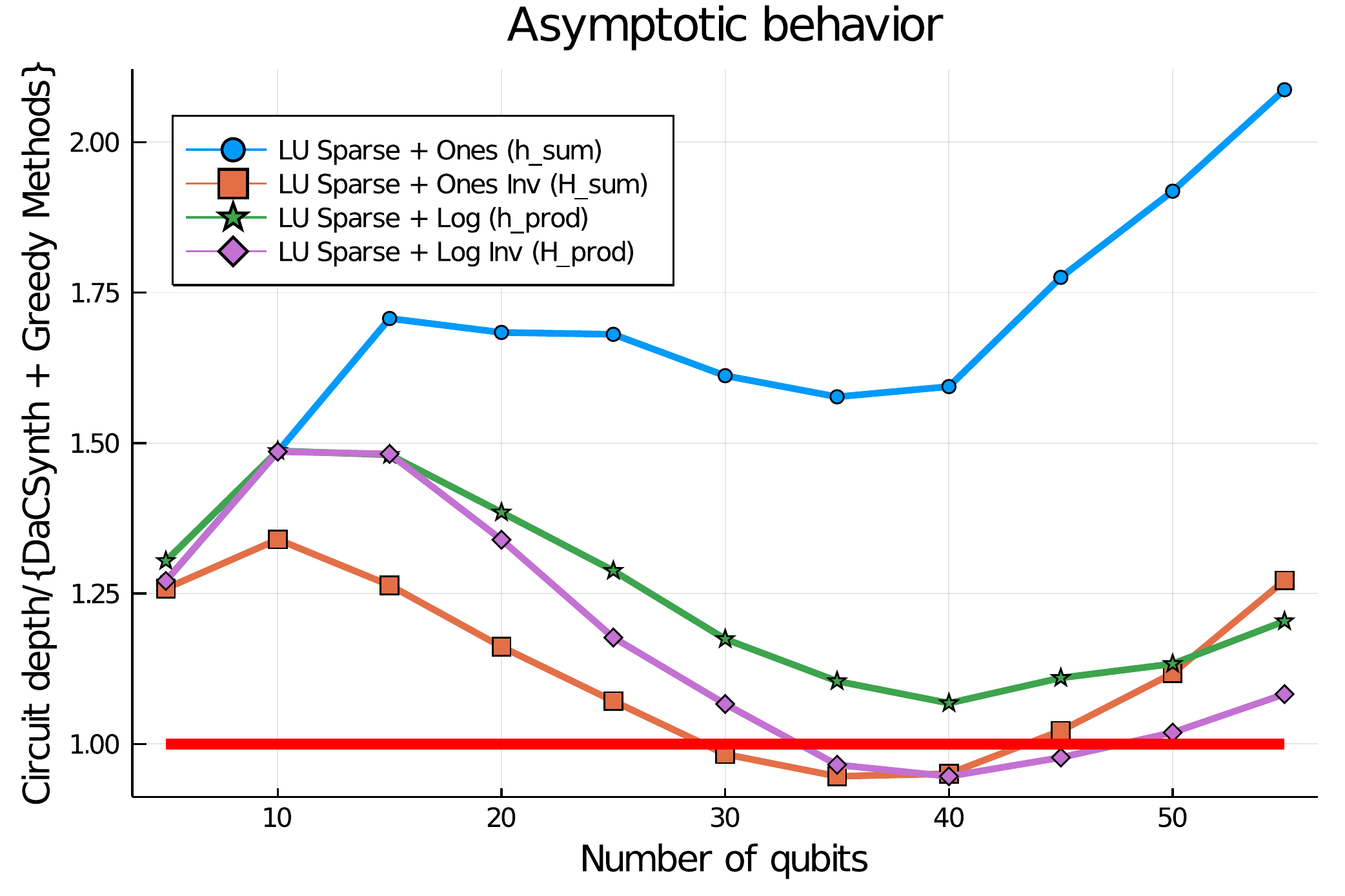}
  \caption{Average performance of \{ LU decomposition + greedy methods \} vs \{ DaCSynth and purely greedy methods \} (without standard deviation).}
  \label{asymptotic_bench_lu_zoom}
\end{subfigure}
\hfill
\begin{subfigure}{0.48\textwidth}
  \includegraphics[scale=0.40]{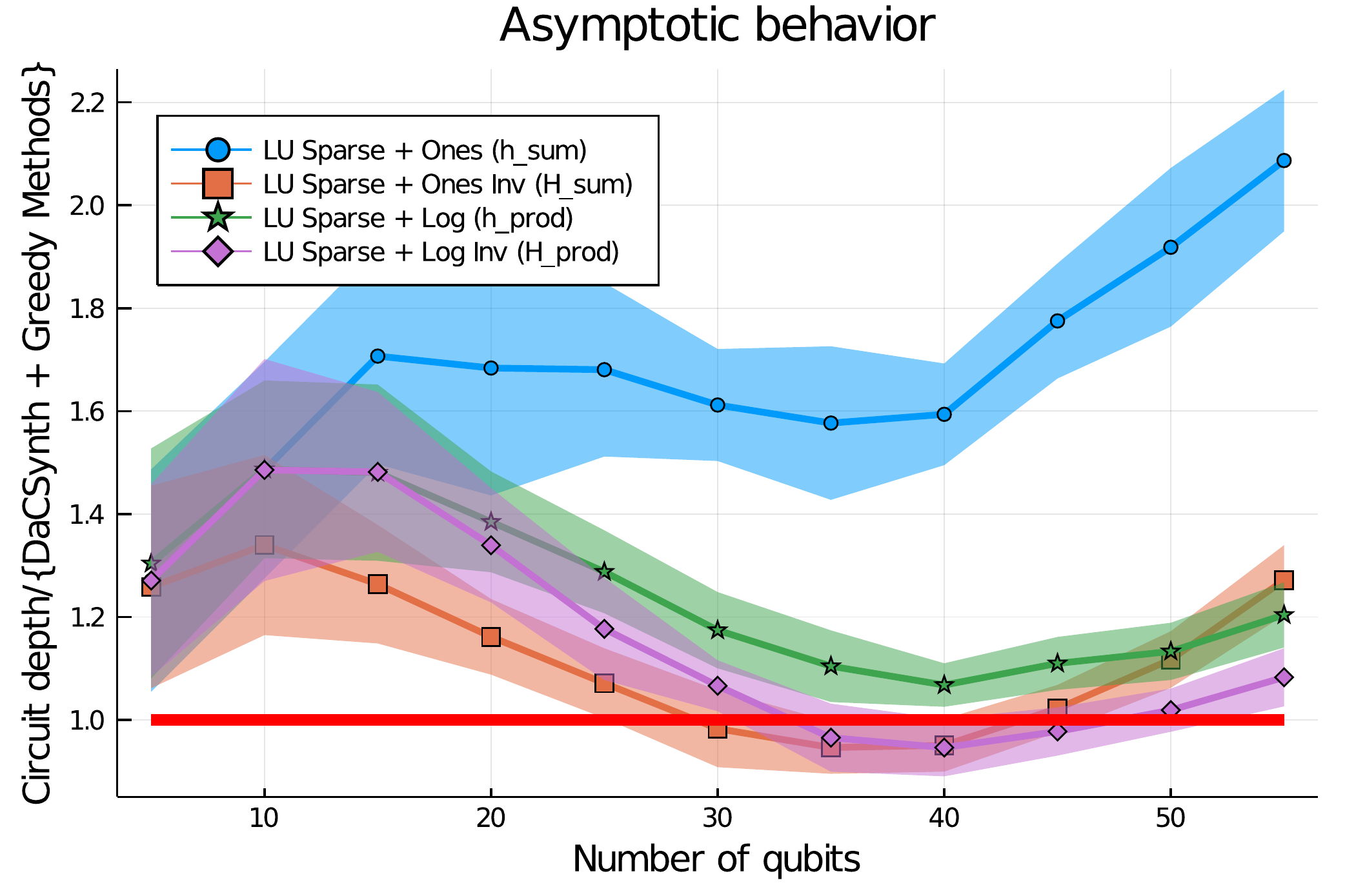}
\caption{Average performance of \{ LU decomposition + greedy methods \} vs \{ DaCSynth and purely greedy methods \} (with standard deviation).}
\label{asymptotic_bench_lu_zoom_std}
\end{subfigure}
\caption{Average performance of \{ LU decomposition + cost minimization techniques \} vs\\ \{ DaCSynth and purely greedy methods \}.}
\label{depth_average_greedy_lu}
\end{figure*}

\subsubsection{Evaluation of the purely greedy algorithms}

We now evaluate the performance of our greedy algorithms against DaCSynth. The results of the worst-case experiment are given in Figure~\ref{asymptotic_bench_greedy} and Figure~\ref{depth_average_greedy}. For clarity we plot the ratio between the depth of the circuits returned by the greedy algorithms and the depth of the circuits returned by DaCSynth. So when the ratio is smaller than 1 this means that the greedy method outperforms DaCSynth. 
The results with the four cost functions are given in Figure
\ref{asymptotic_bench_greedy}. For small $n$ the greedy methods always
outperform DaCSynth but inevitably as $n$ grows the performance of the
greedy methods deteriorates, and this exponentially fast. In fact,
when $n$ is sufficiently large the cost minimization process can no
longer converge to a solution and we stop the experiments when it is
clear that the cost minimization process cannot outperform
DaCSynth. The bad scalability of the greedy methods is particularly
visible with the cost function $h_{\text{sum}}$. Clearly this cost
function is always outperformed by the three others but the scale of
the figure prevents proper discernment of the performance of the other
three cost functions. A zoom in the range $n=0..50$ without the cost
function $h_{\text{sum}}$ is proposed in
Figure~\ref{depth_average_greedy}. In a way, we get results similar to
those obtained for a size optimization approach in
\cite{de2021gaussian}. The $H_{\text{sum}}$ cost function provides the
best results for $n< 25$ but its performances deteriorate faster than
with the log based cost functions. However, contrary to the size
optimization case where we highlighted a slight range of qubits where
the cost function $h_{\text{prod}}$ could perform better than
$H_{\text{sum}}$, here this is the cost function $H_{\text{prod}}$
that seems to give the best results for the approximate range
$n=25..40$. However, when we plot the results with the standard
deviation in Figure~\ref{asymptotic_bench_greedy_zoom_std}, we see
that this advantage of one cost function over the others is relative
in view of the high variance in the results. We recover similar
results with the close-to-optimal experiment, see
Figure~\ref{60qubits_perf}. The two cost functions $H_{\text{sum}}$
and $H_{\text{prod}}$ provides what seems to be optimal results when
the input operator is generated with a shallow circuit. In fact both
cost functions consistently enable to resynthesize the operator with a
shallower circuit than the one given as input. There is a threshold
when the input circuits are of depth $30$ or larger and it becomes
harder for the cost minimization technique to converge promptly to a
solution. The performance of our greedy methods deteriorate with a
high variance in the results, especially with the $H_{\text{sum}}$
cost function, see Figure~\ref{60qubits_std}. Finally when we get
closer and closer to a worst case DaCSynth eventually provides the
best results, in accordance with the worst case experiment done above.

\medskip
We also did some experiments with the greedy algorithms combined with the LU decomposition. We only show the worst-case experiment results in Figure~\ref{depth_average_greedy_lu}. The ratio between the method benchmarked and the best result among the purely greedy methods and DaCSynth is plotted. So if the value of the curve is below $1$ this means that the algorithm outperforms all the other ones we have studied so far. Overall there is a range of qubits (between 30 and 45) where the association of the LU decomposition and the greedy methods slightly provides the best results. Again this advantage must be nuanced because the variance of the results is quite high (see Figure~\ref{asymptotic_bench_lu_zoom_std}). Less significant results were obtained in the close-to-optimal experiment so we do not report them.

\medskip
Overall, for practical synthesis problems, for instance for the library of reversible functions we consider in Section~\ref{sec::bench::reversible}, it is preferable to try all three cost functions, with and without the LU decomposition, and keep the best result. Globally all the execution times required for the greedy methods are in the same range of magnitude, so we only increase the total execution time by a factor given by the number of methods we want to try. Further analysis needs to be done to determine when each method should be used.

\subsection{Benchmarks on reversible functions} \label{sec::bench::reversible}

We apply our method to the synthesis of well known reversible functions. Our goal is to show that we can mitigate the CNOT cost of the circuits while keeping the T-cost as low as possible. Our strategy is simple: we scan a circuit where the T-cost has been optimized and we re-synthesize each CNOT sub-circuit appearing. This way we keep the T-count and the T-depth as low as possible. We already showed in \cite{de2021gaussian} that the CNOT count can be significantly reduced, but also the depth to our surprise. Here we focus primarily on depth optimization.

We choose the Tpar algorithm \cite{DBLP:journals/tcad/AmyMM14} for the pre-processing part: this is the best algorithm to our knowledge for the T-depth optimization. The Tpar algorithm also works with ancillary qubits and requires the synthesis of linear reversible circuits with encoded ancillas where we can use the extension of our framework. 
Since Tpar, other algorithms optimizing the T-count have been proposed \cite{amy2019t,kissinger2020reducing,zhang2019optimizing,2058-9565-4-1-015004} and can be plugged before the Tpar algorithm in the pre-processing part. Even though such algorithms provide better T-counts, the T-par algorithm alone remains competitive for the T-count and we believe using only Tpar does not alter the global message of this section which is that the depth of reversible circuits can be significantly reduced without increasing other metrics of importance like the T-depth. So, for simplicity, we do not consider those newer T-count optimizers.

The library of reversible functions we used is taken from Matthew Amy's github repository \cite{meamy}. Still from from Matthew Amy's github repository we used his C++ implementation of the Tpar algorithm. Although a more recent implementation in Haskell exists, by the time we write this paper it does not take into account ancillary qubits. 

\subsubsection{Ancilla-free results}

Without ancillary qubits, the results are given in Table~\ref{table_results}. For each reversible function we provide the statistics (T-count, T-depth, CNOT count, Total depth) of the original circuit and the circuit optimized solely with the Tpar algorithm. As the T-count and T-depth remains unchanged after our post-optimization process, we only show the new CNOT count and total depth after running our framework for size optimization from \cite{de2021gaussian} and our framework for depth optimization described in this article. For those two metrics (CNOT count and total depth) the savings compared to the Tpar algorithm are also given. 

We also compare ourselves against an heuristic optimization from Nam, Ross, Su, Childs and Maslov \cite{Continuous}, which is now one of the state-of-the-art methods for quantum circuits optimization. Although its primary objective is gate count optimization, the heuristic also improves the total depth of the circuits. Two versions of their algorithm are proposed, corresponding to "light" and "heavy" optimization procedures. The code is not open source but the circuits are available in Neil Ross's github repository \cite{njross}. So, when available, we chose the results of the heavy optimization and reported them in Table~\ref{table_results}.

Compared to Tpar, for almost every function we manage to reduce significantly the total depth: $-47\%$ in average, with a maximum of $-92\%$. The Galois field multipliers are notably the functions that benefit the most from our optimization with at least $65\%$ of reduction. As a bonus, we also have better CNOT counts with $40\%$ of gain in average.

The heuristic from \cite{Continuous} gives the best CNOT counts of the four methods by far. This significant decrease in the CNOT count has also a beneficial impact on the total depth but these gains come at the expense of another metric of importance: the T-depth. More exactly, the heuristic does not benefit from the optimizations done by Tpar, which results in a difference in the T-depth of more than $100\%$ in almost every circuit, with up to $760\%$ for the GF($2^{32})$ multiplier.
Still, we manage to produce circuits with equivalent or even better depths than \cite{Continuous} for $70\%$ of the circuits. Again, our best gains are for the Galois field multipliers with $35\%-40\%$ in average except for GF($2^{32})$ with a gain of $12\%$. There are five circuits for which our method gives significant worse results with an increase of more than $10\%$ in the depth.

Compared to the size optimization framework proposed in \cite{de2021gaussian}, we manage to produce better results for almost every function --- the only exception being the Barenco version of the Toffoli gate on 3 qubits. Again the Galois field multipliers are significantly optimized: for the multiplier in GF($2^{32}$) we reduced the depth from $2130$ to $888$, representing a relative gain of $58\%$. Yet this comes at the price of an increased CNOT count.

\paragraph{}
Overall, these results show that it is possible to significantly optimize the depths of useful circuits while keeping other metrics of importance like the T-depth optimized. We improve a precedent framework that achieves a similar goal but with a focus on size optimization. Compared to a state-of-the-art method that does not optimize the T-depth, we also manage to provide circuits of similar or even better depths, proving that the optimization of the T-depth and the total depth are not incompatible. As future work it would be interesting to design an algorithm that solely focuses on total depth optimization (without particular consideration on the T-depth) and see how our framework compares against.

To evaluate which algorithm was actually useful for providing the best results, we also give in Table~\ref{frequency} the frequency of best performance for each algorithm. More precisely for each method we give the number of times it returned the best result and the number of times it was the only method to return the best result. In \cite{de2021gaussian} we observed that the greedy methods for optimizing the size gave most of the times the best results, emphasizing the idea that overall most of the sub-circuits to re-synthesize correspond to simple operators. We observe a similar pattern here: we could have done the same experiment with solely the greedy algorithms with the three cost functions $H_{\text{sum}}, h_{\text{prod}}$ and $H_{\text{prod}}$. Even though most of the time they all provide the best results, for some specific operators each of these three cost functions were able to uniquely return the shallowest circuits. So we cannot remove one of them. Furthermore the fact that DaCSynth almost never returned the best result is a sign that we never had to synthesize worst case operators. Even for circuits acting on a large number of qubits (the Galois field multiplier GF($2^{32}$) on 96 qubits for instance) the DaCSynth algorithm was not able to back off.

\subsubsection{Results with ancillary qubits}

We now repeat the experiment with the use of ancillary qubits. For each function, we let Tpar compute the number of ancillary qubits necessary to reduce the T-depth to its minimum. Then again we re-synthesize each chunk of purely CNOT circuits. This time due to the use of ancillary qubits we need to synthesize some operators $A \in F_2^{p \times n}$ where $n$ is the number of qubits and $p = n + \# \text{ancillas}$. To do the synthesis of one CNOT circuit we have two options: 
\begin{itemize}
  \item either computing the actual CNOT operator implemented by the CNOT circuit given by the Tpar algorithm, this is the "direct" method,
  \item or we can use our block algorithm described in Section~\ref{sec::extension}, this is the "block" method.
\end{itemize}

Note that even if $p < 2n$ the two methods are not equivalent because in the "block" method we compute the global operator differently (see our explanation in Section~\ref{sec::extension}). 

So for each CNOT sub-circuit we did the synthesis twice, one with each method, and we kept the best result. We had to be careful about the size of the operators to synthesize. For some reversible functions the total number of qubits exceeds several hundreds and some of our methods cannot compute a solution in a reasonable amount of time. So for those extreme cases we had to remove the greedy methods and even the DaCSynth algorithm for the Galois field multiplier in GF($2^{32}$) for which the total number of qubits reaches more than $2500$. Note that these restrictions only concern the "direct" method as the "block" method manages to synthesize only operators on at most $2n$ qubits. We consider it to be an asset of the "block" method over the "direct" synthesis.

The results are given in Table~\ref{table_results_ancilla}. The statistics of the circuit output by the Tpar algorithm are showed and we give the CNOT count and total depth after applying our optimizations. Again the savings compared to the Tpar circuits are given. Overall we have even better savings than in the ancilla-free case. In average we reduce the depth by almost $60\%$. The Galois field multipliers again give the best savings with up to $99\%$ for the multiplier GF($2^{32}$). As a bonus the total number of CNOT gates has also significantly decreased.

The performance frequency of each algorithm is given in Table~\ref{frequency_ancilla}. Again only the greedy methods with the cost functions $H_{\text{sum}}, h_{\text{prod}}$ and $H_{\text{prod}}$ were useful. We also add the number of times the "block" and "direct" methods were the only method to return the best result. When the number of ancillary qubits is small, the direct method has to be privileged. This is probably due to our process to compute the global operator in the "block" method that is not efficient, this should be the subject of a future work. When the number of ancillas increase then the block method is more efficient, notably because the greedy methods do not scale well (both in terms of computational time and circuit size) when the problem size is too large. 

\section{Discussion and future work} \label{sec::discussion}

We see two main areas for improvements: 
\begin{itemize}
  \item First, from a practical point of view, it seems possible to improve the synthesis algorithm when the number of encoded ancillas is not more than $n$. The number of ancillas is not large enough to deploy a block strategy and we proposed a naive way to compute one operator that can "do the job" and that we can directly synthesize. However the benchmarks have revealed that the operator already given by the Tpar algorithm is most of the time a better one. 
  \item Secondly, it seems clear when looking at the benchmarks that the upper bound of DaCSynth is to be improved. It is easy to have some results when we only consider the maximum number of ones in each row of $B = A_3A_1^{-1} \in F_2^{n \times n}$, but combining this with a restriction on the columns makes the proofs harder. For instance it is easy to prove that one can flip no more than $n/2$ entries on each row of $B$ to get a matrix $B \oplus C$ with only two different row vectors. Such matrix can be zeroed in a logarithmic number of steps and this would give an upper bound of approximately $n + \text{logarithmic terms}$ which is closer to our benchmarks. Yet we are unable to get any property on the maximum number of 1 in the columns of $C$ to conclude.
\end{itemize}

We now sketch some ideas of optimizations that could be candidates for improving our framework.
In Section~\ref{sec::extension} we showed that the implementation of an operator $A \in F_2^{p \times n}, p>n$, is in fact as costly as the synthesis of a square operator of size $k \leq 2n$ up to logarithmic terms. 
We want to highlight that this result can lead to new strategies for our initial divide-and-conquer framework and improve the theoretical upper bounds on the depth. 
Consider the matrix $B = A_3A_1^{-1} \in F_2^{n \times n}$ to zero during step 2 of the framework. If $B$ is of rank $k < n$ then we write 
\[ B = DF \]
where $D \in F_2^{n \times k}, F \in F_2^{k \times n}$ are both of rank $k$. By using the block extension algorithm we know that there is a sequence $(E_{i_D j_D})_{i_D j_D}$, resp. $(E_{i_F j_F})_{i_F j_F})$, of row, resp. column, operations of depth $\mathcal{O}(k)$ such that 
\[ \prod_{i_D, j_D} E_{i_D,j_D} D = \begin{pmatrix} I_k \\
    0 \end{pmatrix}\text{\quad and\quad}
 F\prod_{i_F, j_F} E_{i_F,j_F}  = \begin{pmatrix} I_k & 0 \end{pmatrix}.\]
Equivalently 
\[ \prod_{i_D, j_D} E_{i_D,j_D} B \prod_{i_F, j_F} E_{i_F,j_F} = \begin{pmatrix} I_k & 0 \\ 0 & 0 \end{pmatrix}\]
and $B$ can be zeroed with a sequence of operations of depth $\mathcal{O}(k)$. So instead of trying to minimize the number of ones in $B$, as we do, one might be interested in diminishing the rank. The problem can be formulated as: given an integer $r < n$ and $B \in F_2^{n \times n}$ of rank $k > r$, what is the sequence of operations (row operations, column operations, entry flips) of minimum depth that transform $B$ into a matrix of rank $r$?

This problem is related to other problems in the literature. The matrix rigidity of a matrix $A$ is defined as the minimum Hamming distance between $A$ and a matrix of rank $r$. In other words the matrix rigidity of $A$ is the number of entries of A that must be modified in order for the rank to drop to $r$. 
In the literature the concept of matrix rigidity was used to prove lower bounds on the complexity of classical linear circuits \cite{DBLP:conf/mfcs/Valiant77,lokam2009complexity}. Most of the work we found 
on the subject was thus dedicated to finding explicit rigid matrices, which is quite the opposite of our approach. Moreover the distance for the matrix rigidity is defined as the number of flips in $A$ whereas we are concerned in the depth of a sequence of operations. 

The problem of matrix rigidity can be extended with the more general problem of low rank approximations where we try to find, for given target rank $r$,
the solution to 
\[ \underset{\text{rank}(R) = r}{min} \| A - R \|\]
where $\|\cdot \|$ is an appropriate norm. Using the $L_1$ norm and we have the problem of finding the matrix rigidity of $A$. But again 
none of the norms usually considered take into account the depth required to implement $R$. Lastly such problems only consider one of the three
operations that are available to us, namely the entry flips. Although row and column operations alone cannot reduce the rank of a matrix, they can help
in creating a new matrix that needs less entries to flip to have a reduced rank.

Finally can we extend DaCSynth to take into account connectivity constraints? We believe it will be complicated because it is not natural in a restricted
connectivity to split the qubits into two sets, especially if there is no particular symmetry between the two sets. We think it is preferable to improve the results 
from \cite{DBLP:journals/cjtcs/KutinMS07} for the LNN architecture and extend it to an arbitrary topology.

\section{Conclusion} \label{sec::conclusion}

We have proposed DaCSynth, a scalable algorithm for synthesizing shallow linear reversible quantum circuits and we have shown that synthesizing an operator with a divide and conquer algorithm is equivalent to zeroing binary matrices with three elementary operations. This gives a framework that generalizes other works and widens the perspective of finding new techniques.
We have derived an upper bound that improves existing bounds for registers of intermediate sizes and we have used a greedy algorithm to obtain the shallowest possible circuits. In our benchmarks, the circuits produced by the algorithm DaCSynth have a depth which is twice smaller than state-of-the-art algorithms. We have also presented purely greedy algorithms providing close to optimal results for small or simple operators. Applied to the synthesis of a class of reversible functions, we report some substantial savings in the total depth of the circuits while keeping the T-depth as low as possible. 
This work represents one step toward the compilation of quantum circuits that can be executed on a quantum hardware in the future. As future work, we will study how to adapt this method to take into account the connectivity constraints between the qubits in real hardware. Another future work will be to extend DaCSynth to Clifford circuits. Syntheses of Clifford circuits through normal forms are possible (see, e.g., \cite{maslov2018shorter,aaronson2004improved,duncan2020graph}) but it would also be interesting to see if DaCSynth has an equivalent in the symplectic group $Sp(2n, \mathbb{F}_2)$ used to represent Clifford operators with a different set of elementary operations available. An interesting result would be to see if a direct synthesis can produce better depths rather than using normal forms.

\begin{sidewaystable*}
\def\narrow#1{\scalebox{.8}[1.0]{#1}}
\centering
\caption{CNOT optimization of a library of reversible functions with several CNOT circuits synthesis methods. For each reversible function, "Original" reports some statistics (T-count, T-depth, CNOT count, Total depth) of the original circuit, "Tpar" reports the results of the circuit optimized solely by the Tpar algorithm, "Tpar + CNOT size opt." reports the results of the circuit optimized with the Tpar algorithm and post-processed by a size optimization procedure from \cite{de2021gaussian} and finally "Tpar + CNOT depth opt." reports the results of the circuit optimized with the Tpar algorithm and post-processed by our proposed depth optimization procedure. The T-count and T-depth for the last two algorithms are omitted because identical to the Tpar algorithm.}
\label{table_results}
\scalebox{0.7}{
\begin{tabular}{lcc@{~}c@{~}c@{~}cc@{~}c@{~}c@{~}cc@{~}c@{~}c@{~}c@{~}c@{~}c@{~}c@{~}c@{~}cc@{~}c@{~}c@{~}c@{~}cc@{~}c@{~}c@{~}c@{~}c@{~}c} \toprule
      \multirow{2}{*}{Function} & \multirow{2}{*}{\#$n$} &
                                                           \multicolumn{4}{c}{Original} & \multicolumn{4}{c}{Tpar \cite{DBLP:journals/tcad/AmyMM14}} & \multicolumn{9}{c}{Nam \emph{et al.} \cite{Continuous}} & \multicolumn{5}{c}{Tpar + CNOT size opt. \cite{de2021gaussian}} & \multicolumn{6}{c}{Tpar + CNOT depth opt.} \\
      \cmidrule(lr){3-6} \cmidrule(lr){7-10} \cmidrule(lr){11-19} \cmidrule(lr){20-24} \cmidrule(lr){25-30} 
               &     &  T & T & \narrow{CNOT} & \narrow{Depth} &  T & T & \narrow{CNOT} & \narrow{Depth} &
                                                                       T & T & \narrow{\%Diff} & \; & \narrow{CNOT} & \narrow{\%Diff} & \; & \narrow{Depth} & \narrow{\%Diff} & \narrow{CNOT} & \narrow{\%Diff} & \; & \narrow{Depth} & \% \narrow{Diff} & \narrow{CNOT} & \narrow{\%Diff} & \; & \narrow{Depth} & \% \narrow{Diff} & \narrow{\%Diff} \\
               & &\narrow{count} &\narrow{count} & \narrow{count} & ~ & \narrow{count} & \narrow{depth} & \narrow{count} & &
                                     \narrow{count} & \narrow{depth} &
                                                                       \narrow{(vs T-par)} &  & \narrow{count} & \narrow{(vs T-par)} &  & & \narrow{(vs T-par)} & \narrow{count} & \narrow{(vs T-par)} & & & \narrow{(vs T-par)} & \narrow{count} & \narrow{(vs T-par)} & & & \narrow{(vs T-par)} & \narrow{(vs \cite{Continuous})}  \\
      \ \\
Adder\_8& $ 24$ & $399$ & $69$ & $466$ & $223$ & $213$ & $30$ & $741$ & $302$ & $215$ & $73$ & $+143\%$ & & $291$ & $-61\%$ & & $193$ & $-36\%$ & $491$ & $ -34\% $ & & $200$ & $-34\%$ & $511$ & $ -31\% $ & & $190$ & $ -37\% $ & $-2\%$ \\
barenco\_tof\_10& $ 19$ & $224$ & $96$ & $224$ & $288$ & $100$ & $43$ & $332$ & $272$ & $100$ & $81$ & $+88\%$ & & $130$ & $-61\%$ & & $230$ & $-15\%$ & $188$ & $ -43\% $ & & $217$ & $-20\%$ & $188$ & $ -43\% $ & & $215$ & $ -21\% $ & $-7\%$ \\
barenco\_tof\_3& $ 5$ & $28$ & $12$ & $28$ & $36$ & $16$ & $8$ & $52$ & $60$ & $16$ & $16$ & $+100\%$ & & $18$ & $-65\%$ & & $38$ & $-37\%$ & $26$ & $ -50\% $ & & $34$ & $-43\%$ & $26$ & $ -50\% $ & & $35$ & $ -42\% $ & $-8\%$ \\
barenco\_tof\_4& $ 7$ & $56$ & $24$ & $56$ & $72$ & $28$ & $13$ & $96$ & $96$ & $28$ & $27$ & $+108\%$ & & $34$ & $-65\%$ & & $68$ & $-29\%$ & $50$ & $ -48\% $ & & $61$ & $-36\%$ & $50$ & $ -48\% $ & & $61$ & $ -36\% $ & $-10\%$ \\
barenco\_tof\_5& $ 9$ & $84$ & $36$ & $84$ & $108$ & $40$ & $18$ & $134$ & $123$ & $40$ & $36$ & $+100\%$ & & $50$ & $-63\%$ & & $95$ & $-23\%$ & $73$ & $ -46\% $ & & $89$ & $-28\%$ & $73$ & $ -46\% $ & & $87$ & $ -29\% $ & $-8\%$  \\
csla\_mux\_3& $ 15$ & $70$ & $21$ & $90$ & $67$ & $62$ & $8$ & $379$ & $210$ & $64$ & $26$ & $+225\%$ & & $70$ & $-82\%$ & & $63$ & $-70\%$ & $187$ & $ -51\% $ & & $85$ & $-60\%$ & $195$ & $ -49\% $ & & $69$ & $ -67\% $ & $+10\%$ \\
csum\_mux\_9& $ 30$ & $196$ & $18$ & $196$ & $59$ & $84$ & $6$ & $366$ & $153$ & $84$ & $16$ & $+167\%$ & & $140$ & $-62\%$ & & $47$ & $-69\%$ & $179$ & $ -51\% $ & & $76$ & $-50\%$ & $194$ & $ -47\% $ & & $59$ & $ -61\% $ & $+26\%$ \\
cycle\_17\_3& $ 35$ & $4739$ & $2001$ & $4742$ & $5974$ & $1944$ & $562$ & $6608$ & $5231$ & - & - & - & & - & - & & - & - & $4267$ & $ -35\% $ & & $4215$ & $-19\%$ & $4810$ & $ -27\% $ & & $3884$ & $ -26\% $ & - \\
GF($2^{10}$)\_mult& $ 30$ & $700$ & $108$ & $709$ & $290$ & $410$ & $16$ & $2206$ & $1026$ & $410$ & $117$ & $+631\%$ & & $609$ & $-72\%$ & & $291$ & $-72\%$ & $949$ & $ -57\% $ & & $299$ & $-71\%$ & $1124$ & $ -49\% $ & & $192$ & $ -81\% $ & $-34\%$ \\
GF($2^{16}$)\_mult& $ 48$ & $1792$ & $180$ & $1837$ & $489$ & $1040$ & $24$ & $6724$ & $2551$ & $1040$ & $196$ & $+717\%$ & & $1581$ & $-76\%$ & & $492$ & $-81\%$ & $2545$ & $ -62\% $ & & $687$ & $-73\%$ & $3314$ & $ -51\% $ & & $345$ & $ -86\% $ & $-30\%$ \\
GF($2^{32}$)\_mult& $ 96$ & $7168$ & $372$ & $7292$ & $1001$ & $4128$ & $47$ & $34244$ & $11520$ & $4128$ & $404$ & $+760\%$ & & $6299$ & $-82\%$ & & $1006$ & $-91\%$ &  $10972$ & $ -68\% $ & & $2130$ & $-82\%$ & $15680$ & $ -54\% $ & & $888$ & $ -92\% $ & $-12\%$ \\
GF($2^4$)\_mult& $ 12$ & $112$ & $36$ & $115$ & $99$ & $68$ & $6$ & $307$ & $173$ & $68$ & $40$ & $+567\%$ & & $99$ & $-68\%$ & & $102$ & $-41\%$ &  $135$ & $ -56\% $ & & $77$ & $-55\%$ & $154$ & $ -50\% $ & & $61$ & $ -65\% $ & $-40\%$ \\
GF($2^5$)\_mult& $ 15$ & $175$ & $48$ & $179$ & $130$ & $115$ & $9$ & $502$ & $259$ & $115$ & $53$ & $+489\%$ & & $154$ & $-69\%$ & & $131$ & $-49\%$ &  $210$ & $ -58\% $ & & $117$ & $-55\%$ & $236$ & $ -53\% $ & & $80$ & $ -69\% $ & $-39\%$ \\
GF($2^6$)\_mult& $ 18$ & $252$ & $60$ & $257$ & $163$ & $150$ & $9$ & $660$ & $350$ & $150$ & $66$ & $+633\%$ & & $221$ & $-67\%$ & & $166$ & $-53\%$ &  $308$ & $ -53\% $ & & $134$ & $-62\%$ & $372$ & $ -44\% $ & & $95$ & $ -73\% $ & $-43\%$ \\
GF($2^7$)\_mult& $ 21$ & $343$ & $72$ & $349$ & $195$ & $217$ & $12$ & $996$ & $490$ & $217$ & $80$ & $+567\%$ & & $300$ & $-70\%$ & & $198$ & $-60\%$ &  $442$ & $ -56\% $ & & $178$ & $-64\%$ & $524$ & $ -47\% $ & & $127$ & $ -74\% $ & $-36\%$ \\
GF($2^8$)\_mult& $ 24$ & $448$ & $84$ & $469$ & $233$ & $264$ & $13$ & $1254$ & $619$ & $264$ & $92$ & $+608\%$ & & $405$ & $-68\%$ & & $236$ & $-62\%$ &  $589$ & $ -53\% $ & & $223$ & $-64\%$ & $691$ & $ -45\% $ & & $146$ & $ -76\% $ & $-38\%$ \\
GF($2^9$)\_mult& $ 27$ & $567$ & $96$ & $575$ & $258$ & $351$ & $15$ & $1712$ & $810$ & $351$ & $105$ & $+600\%$ & & $494$ & $-71\%$ & & $259$ & $-68\%$ &  $753$ & $ -56\% $ & & $272$ & $-66\%$ & $894$ & $ -48\% $ & & $160$ & $ -80\% $ & $-38\%$ \\
grover\_5& $ 9$ & $336$ & $144$ & $336$ & $457$ & $154$ & $51$ & $499$ & $477$ & - & - & - & & - & - & & - & - &  $331$ & $ -34\% $ & & $377$ & $-21\%$ & $332$ & $ -33\% $ & & $370$ & $ -22\% $ & - \\
ham15-high& $ 20$ & $2457$ & $996$ & $2500$ & $3026$ & $1019$ & $380$ & $3427$ & $2956$ & - & - & - & & - & - & & - & - &  $2183$ & $ -36\% $ & & $2227$ & $-25\%$ & $2261$ & $ -34\% $ & & $2153$ & $ -27\% $ & - \\
ham15-low& $ 17$ & $161$ & $69$ & $259$ & $263$ & $97$ & $33$ & $471$ & $360$ & - & - & - & & - & - & & - & - &  $280$ & $ -41\% $ & & $221$ & $-39\%$ & $294$ & $ -38\% $ & & $206$ & $ -43\% $ & - \\
ham15-med& $ 17$ & $574$ & $240$ & $616$ & $750$ & $230$ & $84$ & $759$ & $682$ & - & - & - & & - & - & & - & - &  $481$ & $ -37\% $ & & $503$ & $-26\%$ & $492$ & $ -35\% $ & & $490$ & $ -28\% $ & - \\
hwb6& $ 7$ & $105$ & $45$ & $131$ & $152$ & $75$ & $24$ & $270$ & $248$ & - & - & - & & - & - & & - & - & $172$ & $ -36\% $ & & $157$ & $-37\%$ & $177$ & $ -34\% $ & & $153$ & $ -38\% $ & - \\
hwb8& $ 12$ & $5887$ & $2139$ & $7970$ & $7956$ & $3531$ & $860$ & $22670$ & $15838$ & - & - & - & & - & - & & - & - &  $13351$ & $ -41\% $ & & $9120$ & $-42\%$ & $14226$ & $ -37\% $ & & $8247$ & $ -48\% $ & - \\
mod5\_4& $ 5$ & $28$ & $12$ & $32$ & $41$ & $16$ & $6$ & $48$ & $57$ & $16$ & $14$ & $+133\%$ & & $28$ & $-42\%$ & & $46$ & $-19\%$ & $32$ & $ -33\% $ & & $39$ & $-32\%$ & $33$ & $ -31\% $ & & $40$ & $ -30\% $ & $-13\%$ \\
mod\_adder\_1024& $ 28$ & $1995$ & $831$ & $2005$ & $2503$ & $1011$ & $258$ & $3650$ & $2560$ & $1011$ & $743$ & $+188\%$ & & $1278$ & $-65\%$ & & $2086$ & $-19\%$ &  $2369$ & $ -35\% $ & & $1906$ & $-26\%$ & $2575$ & $ -29\% $ & & $1755$ & $ -31\% $ & $-16\%$ \\
mod\_adder\_1048576& $ 58$ & $17290$ & $7292$ & $17310$ & $21807$ & $7298$ & $1927$ & $29794$ & $19975$ & - & - & - & & - & - & & - & - & $19122$ & $ -36\% $ & & $15258$ & $-24\%$ & $21179$ & $ -29\% $ & & $14055$ & $ -30\% $ \\
mod\_mult\_55& $ 9$ & $49$ & $15$ & $55$ & $50$ & $35$ & $7$ & $106$ & $75$ & $35$ & $19$ & $+171\%$ & & $40$ & $-62\%$ & & $47$ & $-37\%$ & $73$ & $ -31\% $ & & $52$ & $-31\%$ & $75$ & $ -29\% $ & & $50$ & $ -33\% $ & $+6\%$ \\
mod\_red\_21& $ 11$ & $119$ & $48$ & $122$ & $158$ & $73$ & $25$ & $223$ & $207$ & $73$ & $50$ & $+100\%$ & & $77$ & $-65\%$ & & $129$ & $-38\%$ & $136$ & $ -39\% $ & & $145$ & $-30\%$ & $139$ & $ -38\% $ & & $137$ & $ -34\% $ & $+6\%$ \\
qcla\_adder\_10& $ 36$ & $238$ & $24$ & $267$ & $73$ & $162$ & $13$ & $648$ & $195$ & $162$ & $26$ & $+100\%$ & & $183$ & $-72\%$ & & $63$ & $-68\%$ & $363$ & $ -44\% $ & & $92$ & $-53\%$ & $375$ & $ -42\% $ & & $85$ & $ -56\% $ & $+35\%$ \\
qcla\_com\_7& $ 24$ & $203$ & $27$ & $215$ & $81$ & $94$ & $12$ & $371$ & $154$ & $95$ & $26$ & $+117\%$ & & $132$ & $-64\%$ & & $69$ & $-55\%$ & $202$ & $ -46\% $ & & $82$ & $-47\%$ & $208$ & $ -44\% $ & & $77$ & $ -50\% $ & $+12\%$ \\
qcla\_mod\_7& $ 26$ & $413$ & $66$ & $441$ & $197$ & $231$ & $28$ & $813$ & $296$ & $235$ & $72$ & $+157\%$ & & $292$ & $-64\%$ & & $183$ & $-38\%$ & $479$ & $ -41\% $ & & $176$ & -41\% & $494$ & $ -39\% $ & & $169$ & $ -43\% $ & $-8\%$ \\
qft\_4& $ 5$ & $69$ & $48$ & $48$ & $142$ & $67$ & $44$ & $96$ & $185$ & - & - & - & & - & - & & - & - & $56$ & $ -42\% $ & & $150$ & $-19\%$ & $56$ & $ -42\% $ & & $149$ & $ -19\% $ & - \\
rc\_adder\_6& $ 14$ & $77$ & $33$ & $104$ & $104$ & $47$ & $22$ & $165$ & $157$ & $47$ & $31$ & $+41\%$ & & $71$ & $-57\%$ & & $83$ & $-47\%$ & $100$ & $ -39\% $ & & $103$ & $-34\%$ & $100$ & $ -39\% $ & & $95$ & $ -39\% $ & $+14\%$ \\
tof\_10& $ 19$ & $119$ & $51$ & $119$ & $153$ & $71$ & $27$ & $236$ & $190$ & $71$ & $55$ & $+104\%$ & & $70$ & $-70\%$ & & $136$ & $-28\%$ & $130$ & $ -45\% $ & & $140$ & -26\% & $132$ & $ -44\% $ & & $136$ & $ -28\% $ & $0\%$ \\
tof\_3& $ 5$ & $21$ & $9$ & $21$ & $27$ & $15$ & $6$ & $35$ & $46$ & $15$ & $13$ & $+117\%$ & & $14$ & $-60\%$ & & $31$ & $-33\%$ & $21$ & $ -40\% $ & & $31$ & $-33\%$ & $21$ & $ -40\% $ & & $31$ & $ -33\% $ & $0\%$ \\
tof\_4& $ 7$ & $35$ & $15$ & $35$ & $45$ & $23$ & $9$ & $63$ & $71$ & $23$ & $19$ & $+111\%$ & & $22$ & $-65\%$ & & $46$ & $-35\%$ & $37$ & $ -41\% $ & & $49$ & $-31\%$ & $39$ & $ -38\% $ & & $43$ & $ -39\% $ & $-7\%$ \\
tof\_5& $ 9$ & $49$ & $21$ & $49$ & $63$ & $31$ & $12$ & $97$ & $104$ & $31$ & $25$ & $+108\%$ & & $30$ & $-69\%$ & & $61$ & $-41\%$ & $50$ & $ -48\% $ & & $67$ & $-36\%$ & $52$ & $ -46\% $ & & $63$ & $ -39\% $ & $+3\%$ \\
vbe\_adder\_3& $ 10$ & $70$ & $24$ & $80$ & $79$ & $24$ & $9$ & $120$ & $88$ & $24$ & $16$ & $+78\%$ & & $50$ & $-58\%$ & & $56$ & $-36\%$ & $61$ & $ -49\% $ & & $46$ & $-48\%$ & $61$ & $ -49\% $ & & $45$ & $ -49\% $ & $-20\%$ \\
\midrule
      Mean difference    &&&&&&&&&&&&$+276.8\%$&&&  $-66\%$  &&&$-46.55\%$&&$-45.03\%$&&&$-41.66\%$&&$-41.37\%$&&&$-46.68\%$&$-10.24\%$\\
      Worst savings &&&&&&&&&&&&$+760\%$ &&&  $-42\%$  &&&$-15\%$&&$-31\%$   &&&-19\%&&$-27\%$&&&$-19\%$&$+35\%$\\
      Best savings &&&&&&&&&&&&$+41\%$   &&&  $-82\%$  &&&$-91\%$&&$-68\%$   &&&-82\%&&$-54\%$&&&$-92\%$&$-43\%$\\
      \bottomrule
      \end{tabular}
}
\end{sidewaystable*}

\begin{sidewaystable*}
\def\narrow#1{\scalebox{.8}[1.0]{#1}}
  \centering
\caption{Frequency of best performance of each algorithm during the optimization of reversible circuits. For each algorithm, the first column gives the number of times it has returned the best result (possibly other algorithms returned circuits of same size). The second column reports the number of times it was the only one to provide the best possible circuit.}\label{frequency}
\scalebox{0.7}{
\begin{tabular}{l@{~}c@{~}c@{~}c@{~}c@{~}c@{~}c@{~}c@{~}c@{~}c@{~}c@{~}c@{~}c@{~}c@{~}c@{~}c@{~}c@{~}c@{~}c@{~}c@{~}c}
  \toprule
  \multirow{2}{*}{Function} & \multirow{2}{*}{\#$n$} & \multirow{2}{*}{\makecell{\#CNOT \\ sub-circuits}} & \multicolumn{2}{c}{Kutin et al} & \multicolumn{2}{c}{DaCSynth} & \multicolumn{2}{c}{Greedy ($H_{\text{sum}}$, size)} & \multicolumn{2}{c}{Greedy ($h_{\text{prod}}$, size)} & \multicolumn{2}{c}{Greedy ($H_{\text{sum}}$)} & \multicolumn{2}{c}{Greedy ($h_{\text{prod}}$)} & \multicolumn{2}{c}{Greedy ($H_{\text{prod}}$)} & \multicolumn{2}{c}{LU + Greedy ($H_{\text{sum}}$)} & \multicolumn{2}{c}{LU + Greedy ($H_{\text{prod}}$)} \\
  \cmidrule(lr){4-5} \cmidrule(lr){6-7} \cmidrule(lr){8-9} \cmidrule(lr){10-11} \cmidrule(lr){12-13} \cmidrule(lr){14-15} \cmidrule(lr){16-17} \cmidrule(lr){18-19} \cmidrule(lr){20-21}
  & & & Best  & Only  & Best  & Only & Best   & Only  & Best  & Only  & Best   & Only  & Best  & Only  & Best  & Only  & Best  & Only  & Best  & Only  \\
Adder\_8& $ 24$ & 60 & 29 (48\%) & 0 (0\%)& 15 (25\%) & 0 (0\%)& 39 (65\%) & 0 (0\%)& 38 (63\%) & 0 (0\%)& 60 (100\%) & 3 (5\%)& 53 (88\%) & 0 (0\%)& 53 (88\%) & 0 (0\%)& 37 (62\%) & 0 (0\%)& 38 (63\%) & 0 (0\%)\\
barenco\_tof\_10& $ 19$ & 99 & 64 (65\%) & 0 (0\%)& 44 (44\%) & 0 (0\%)& 99 (100\%) & 0 (0\%)& 99 (100\%) & 0 (0\%)& 99 (100\%) & 0 (0\%)& 99 (100\%) & 0 (0\%)& 99 (100\%) & 0 (0\%)& 83 (84\%) & 0 (0\%)& 83 (84\%) & 0 (0\%)\\
barenco\_tof\_3& $ 5$ & 14 & 11 (79\%) & 0 (0\%)& 7 (50\%) & 0 (0\%)& 14 (100\%) & 0 (0\%)& 14 (100\%) & 0 (0\%)& 14 (100\%) & 0 (0\%)& 14 (100\%) & 0 (0\%)& 14 (100\%) & 0 (0\%)& 12 (86\%) & 0 (0\%)& 12 (86\%) & 0 (0\%)\\
barenco\_tof\_4& $ 7$ & 27 & 20 (74\%) & 0 (0\%)& 15 (56\%) & 0 (0\%)& 27 (100\%) & 0 (0\%)& 27 (100\%) & 0 (0\%)& 27 (100\%) & 0 (0\%)& 27 (100\%) & 0 (0\%)& 27 (100\%) & 0 (0\%)& 24 (89\%) & 0 (0\%)& 24 (89\%) & 0 (0\%)\\
barenco\_tof\_5& $ 9$ & 39 & 27 (69\%) & 0 (0\%)& 20 (51\%) & 0 (0\%)& 39 (100\%) & 0 (0\%)& 39 (100\%) & 0 (0\%)& 39 (100\%) & 0 (0\%)& 39 (100\%) & 0 (0\%)& 39 (100\%) & 0 (0\%)& 32 (82\%) & 0 (0\%)& 32 (82\%) & 0 (0\%)\\
csla\_mux\_3& $ 15$ & 15 & 1 (7\%) & 0 (0\%)& 1 (7\%) & 0 (0\%)& 6 (40\%) & 0 (0\%)& 2 (13\%) & 0 (0\%)& 13 (87\%) & 1 (7\%)& 12 (80\%) & 2 (13\%)& 12 (80\%) & 0 (0\%)& 5 (33\%) & 0 (0\%)& 5 (33\%) & 0 (0\%)\\
csum\_mux\_9& $ 30$ & 14 & 5 (36\%) & 0 (0\%)& 2 (14\%) & 0 (0\%)& 5 (36\%) & 0 (0\%)& 5 (36\%) & 0 (0\%)& 12 (86\%) & 4 (29\%)& 9 (64\%) & 2 (14\%)& 7 (50\%) & 0 (0\%)& 5 (36\%) & 0 (0\%)& 5 (36\%) & 0 (0\%)\\
cycle\_17\_3& $ 35$ & 1436 & 725 (50\%) & 0 (0\%)& 433 (30\%) & 0 (0\%)& 871 (61\%) & 0 (0\%)& 840 (58\%) & 0 (0\%)& 1417 (99\%) & 4 (0\%)& 1211 (84\%) & 12 (1\%)& 1385 (96\%) & 1 (0\%)& 910 (63\%) & 0 (0\%)& 913 (64\%) & 0 (0\%)\\
GF($2^{10}$)\_mult& $ 30$ & 20 & 1 (5\%) & 0 (0\%)& 1 (5\%) & 0 (0\%)& 2 (10\%) & 0 (0\%)& 2 (10\%) & 0 (0\%)& 16 (80\%) & 6 (30\%)& 8 (40\%) & 1 (5\%)& 11 (55\%) & 2 (10\%)& 1 (5\%) & 0 (0\%)& 1 (5\%) & 0 (0\%)\\
GF($2^{16}$)\_mult& $ 48$ & 28 & 1 (4\%) & 0 (0\%)& 1 (4\%) & 0 (0\%)& 1 (4\%) & 0 (0\%)& 1 (4\%) & 0 (0\%)& 17 (61\%) & 6 (21\%)& 6 (21\%) & 3 (11\%)& 18 (64\%) & 7 (25\%)& 1 (4\%) & 0 (0\%)& 1 (4\%) & 0 (0\%)\\
GF($2^{32}$)\_mult& $ 96$ & 51 & 1 (2\%) & 0 (0\%)& 1 (2\%) & 0 (0\%)& 1 (2\%) & 0 (0\%)& 1 (2\%) & 0 (0\%)& 13 (25\%) & 5 (10\%)& 10 (20\%) & 6 (12\%)& 40 (78\%) & 30 (59\%)& 1 (2\%) & 0 (0\%)& 1 (2\%) & 0 (0\%)\\
GF($2^4$)\_mult& $ 12$ & 10 & 2 (20\%) & 0 (0\%)& 1 (10\%) & 0 (0\%)& 3 (30\%) & 0 (0\%)& 2 (20\%) & 0 (0\%)& 8 (80\%) & 1 (10\%)& 7 (70\%) & 1 (10\%)& 8 (80\%) & 0 (0\%)& 2 (20\%) & 0 (0\%)& 2 (20\%) & 0 (0\%)\\
GF($2^5$)\_mult& $ 15$ & 13 & 1 (8\%) & 0 (0\%)& 2 (15\%) & 0 (0\%)& 3 (23\%) & 0 (0\%)& 2 (15\%) & 0 (0\%)& 12 (92\%) & 4 (31\%)& 5 (38\%) & 0 (0\%)& 8 (62\%) & 1 (8\%)& 3 (23\%) & 0 (0\%)& 3 (23\%) & 0 (0\%)\\
GF($2^6$)\_mult& $ 18$ & 13 & 2 (15\%) & 0 (0\%)& 2 (15\%) & 0 (0\%)& 1 (8\%) & 0 (0\%)& 1 (8\%) & 0 (0\%)& 10 (77\%) & 5 (38\%)& 6 (46\%) & 0 (0\%)& 8 (62\%) & 1 (8\%)& 2 (15\%) & 0 (0\%)& 2 (15\%) & 0 (0\%)\\
GF($2^7$)\_mult& $ 21$ & 16 & 4 (25\%) & 0 (0\%)& 1 (6\%) & 0 (0\%)& 1 (6\%) & 0 (0\%)& 1 (6\%) & 0 (0\%)& 14 (88\%) & 5 (31\%)& 6 (38\%) & 0 (0\%)& 10 (62\%) & 1 (6\%)& 4 (25\%) & 0 (0\%)& 4 (25\%) & 0 (0\%)\\
GF($2^8$)\_mult& $ 24$ & 17 & 2 (12\%) & 0 (0\%)& 1 (6\%) & 0 (0\%)& 2 (12\%) & 0 (0\%)& 2 (12\%) & 0 (0\%)& 16 (94\%) & 9 (53\%)& 7 (41\%) & 0 (0\%)& 7 (41\%) & 0 (0\%)& 2 (12\%) & 0 (0\%)& 2 (12\%) & 0 (0\%)\\
GF($2^9$)\_mult& $ 27$ & 19 & 2 (11\%) & 0 (0\%)& 1 (5\%) & 0 (0\%)& 3 (16\%) & 0 (0\%)& 3 (16\%) & 0 (0\%)& 15 (79\%) & 6 (32\%)& 8 (42\%) & 2 (11\%)& 9 (47\%) & 2 (11\%)& 4 (21\%) & 0 (0\%)& 3 (16\%) & 0 (0\%)\\
grover\_5& $ 9$ & 123 & 85 (69\%) & 0 (0\%)& 55 (45\%) & 0 (0\%)& 113 (92\%) & 0 (0\%)& 111 (90\%) & 0 (0\%)& 123 (100\%) & 0 (0\%)& 122 (99\%) & 0 (0\%)& 122 (99\%) & 0 (0\%)& 102 (83\%) & 0 (0\%)& 101 (82\%) & 0 (0\%)\\
ham15-high& $ 20$ & 852 & 527 (62\%) & 0 (0\%)& 343 (40\%) & 0 (0\%)& 746 (88\%) & 0 (0\%)& 727 (85\%) & 0 (0\%)& 845 (99\%) & 5 (1\%)& 831 (98\%) & 4 (0\%)& 833 (98\%) & 1 (0\%)& 653 (77\%) & 0 (0\%)& 657 (77\%) & 0 (0\%)\\
ham15-low& $ 17$ & 69 & 43 (62\%) & 0 (0\%)& 30 (43\%) & 0 (0\%)& 50 (72\%) & 0 (0\%)& 49 (71\%) & 0 (0\%)& 64 (93\%) & 1 (1\%)& 67 (97\%) & 3 (4\%)& 61 (88\%) & 0 (0\%)& 47 (68\%) & 0 (0\%)& 47 (68\%) & 0 (0\%)\\
ham15-med& $ 17$ & 189 & 135 (71\%) & 0 (0\%)& 84 (44\%) & 0 (0\%)& 173 (92\%) & 0 (0\%)& 168 (89\%) & 0 (0\%)& 187 (99\%) & 1 (1\%)& 187 (99\%) & 1 (1\%)& 186 (98\%) & 0 (0\%)& 163 (86\%) & 0 (0\%)& 165 (87\%) & 0 (0\%)\\
hwb6& $ 7$ & 48 & 29 (60\%) & 0 (0\%)& 16 (33\%) & 0 (0\%)& 36 (75\%) & 0 (0\%)& 31 (65\%) & 0 (0\%)& 48 (100\%) & 1 (2\%)& 45 (94\%) & 0 (0\%)& 44 (92\%) & 0 (0\%)& 34 (71\%) & 0 (0\%)& 31 (65\%) & 0 (0\%)\\
hwb8& $ 12$ & 2130 & 523 (25\%) & 0 (0\%)& 368 (17\%) & 1 (0\%)& 1084 (51\%) & 2 (0\%)& 913 (43\%) & 1 (0\%)& 2048 (96\%) & 109 (5\%)& 1835 (86\%) & 44 (2\%)& 1792 (84\%) & 4 (0\%)& 1005 (47\%) & 2 (0\%)& 961 (45\%) & 0 (0\%)\\
mod5\_4& $ 5$ & 14 & 9 (64\%) & 0 (0\%)& 10 (71\%) & 0 (0\%)& 13 (93\%) & 0 (0\%)& 14 (100\%) & 0 (0\%)& 14 (100\%) & 0 (0\%)& 14 (100\%) & 0 (0\%)& 13 (93\%) & 0 (0\%)& 14 (100\%) & 0 (0\%)& 14 (100\%) & 0 (0\%)\\
mod\_adder\_1024& $ 28$ & 716 & 372 (52\%) & 0 (0\%)& 243 (34\%) & 0 (0\%)& 482 (67\%) & 0 (0\%)& 469 (66\%) & 0 (0\%)& 693 (97\%) & 0 (0\%)& 664 (93\%) & 4 (1\%)& 662 (92\%) & 1 (0\%)& 511 (71\%) & 0 (0\%)& 507 (71\%) & 0 (0\%)\\
mod\_adder\_1048576& $ 58$ & 5615 & 2695 (48\%) & 0 (0\%)& 1581 (28\%) & 0 (0\%)& 3527 (63\%) & 0 (0\%)& 3444 (61\%) & 0 (0\%)& 5450 (97\%) & 13 (0\%)& 4926 (88\%) & 74 (1\%)& 5137 (91\%) & 0 (0\%)& 3635 (65\%) & 0 (0\%)& 3610 (64\%) & 0 (0\%)\\
mod\_mult\_55& $ 9$ & 14 & 6 (43\%) & 0 (0\%)& 5 (36\%) & 0 (0\%)& 7 (50\%) & 0 (0\%)& 6 (43\%) & 0 (0\%)& 13 (93\%) & 0 (0\%)& 11 (79\%) & 0 (0\%)& 13 (93\%) & 1 (7\%)& 6 (43\%) & 0 (0\%)& 6 (43\%) & 0 (0\%)\\
mod\_red\_21& $ 11$ & 46 & 31 (67\%) & 0 (0\%)& 19 (41\%) & 0 (0\%)& 39 (85\%) & 0 (0\%)& 36 (78\%) & 0 (0\%)& 46 (100\%) & 0 (0\%)& 43 (93\%) & 0 (0\%)& 44 (96\%) & 0 (0\%)& 38 (83\%) & 0 (0\%)& 38 (83\%) & 0 (0\%)\\
qcla\_adder\_10& $ 36$ & 24 & 10 (42\%) & 0 (0\%)& 4 (17\%) & 0 (0\%)& 15 (62\%) & 0 (0\%)& 10 (42\%) & 0 (0\%)& 23 (96\%) & 1 (4\%)& 19 (79\%) & 0 (0\%)& 23 (96\%) & 0 (0\%)& 18 (75\%) & 0 (0\%)& 16 (67\%) & 0 (0\%)\\
qcla\_com\_7& $ 24$ & 26 & 10 (38\%) & 0 (0\%)& 5 (19\%) & 0 (0\%)& 18 (69\%) & 0 (0\%)& 19 (73\%) & 0 (0\%)& 25 (96\%) & 1 (4\%)& 24 (92\%) & 1 (4\%)& 23 (88\%) & 0 (0\%)& 14 (54\%) & 0 (0\%)& 13 (50\%) & 0 (0\%)\\
qcla\_mod\_7& $ 26$ & 52 & 19 (37\%) & 0 (0\%)& 8 (15\%) & 0 (0\%)& 38 (73\%) & 0 (0\%)& 37 (71\%) & 0 (0\%)& 51 (98\%) & 1 (2\%)& 45 (87\%) & 0 (0\%)& 49 (94\%) & 1 (2\%)& 33 (63\%) & 0 (0\%)& 30 (58\%) & 0 (0\%)\\
qft\_4& $ 5$ & 29 & 24 (83\%) & 0 (0\%)& 17 (59\%) & 0 (0\%)& 29 (100\%) & 0 (0\%)& 29 (100\%) & 0 (0\%)& 29 (100\%) & 0 (0\%)& 29 (100\%) & 0 (0\%)& 29 (100\%) & 0 (0\%)& 28 (97\%) & 0 (0\%)& 28 (97\%) & 0 (0\%)\\
rc\_adder\_6& $ 14$ & 49 & 41 (84\%) & 0 (0\%)& 29 (59\%) & 0 (0\%)& 49 (100\%) & 0 (0\%)& 49 (100\%) & 0 (0\%)& 49 (100\%) & 0 (0\%)& 49 (100\%) & 0 (0\%)& 49 (100\%) & 0 (0\%)& 43 (88\%) & 0 (0\%)& 43 (88\%) & 0 (0\%)\\
tof\_10& $ 19$ & 52 & 41 (79\%) & 0 (0\%)& 22 (42\%) & 0 (0\%)& 42 (81\%) & 0 (0\%)& 41 (79\%) & 0 (0\%)& 52 (100\%) & 0 (0\%)& 51 (98\%) & 0 (0\%)& 51 (98\%) & 0 (0\%)& 45 (87\%) & 0 (0\%)& 47 (90\%) & 0 (0\%)\\
tof\_3& $ 5$ & 12 & 11 (92\%) & 0 (0\%)& 8 (67\%) & 0 (0\%)& 12 (100\%) & 0 (0\%)& 12 (100\%) & 0 (0\%)& 12 (100\%) & 0 (0\%)& 12 (100\%) & 0 (0\%)& 12 (100\%) & 0 (0\%)& 11 (92\%) & 0 (0\%)& 11 (92\%) & 0 (0\%)\\
tof\_4& $ 7$ & 17 & 15 (88\%) & 0 (0\%)& 10 (59\%) & 0 (0\%)& 15 (88\%) & 0 (0\%)& 15 (88\%) & 0 (0\%)& 17 (100\%) & 0 (0\%)& 16 (94\%) & 0 (0\%)& 16 (94\%) & 0 (0\%)& 16 (94\%) & 0 (0\%)& 16 (94\%) & 0 (0\%)\\
tof\_5& $ 9$ & 22 & 17 (77\%) & 0 (0\%)& 11 (50\%) & 0 (0\%)& 19 (86\%) & 0 (0\%)& 19 (86\%) & 0 (0\%)& 22 (100\%) & 0 (0\%)& 21 (95\%) & 0 (0\%)& 22 (100\%) & 0 (0\%)& 18 (82\%) & 0 (0\%)& 18 (82\%) & 0 (0\%)\\
vbe\_adder\_3& $ 10$ & 19 & 12 (63\%) & 0 (0\%)& 9 (47\%) & 0 (0\%)& 18 (95\%) & 0 (0\%)& 16 (84\%) & 0 (0\%)& 19 (100\%) & 0 (0\%)& 19 (100\%) & 0 (0\%)& 19 (100\%) & 0 (0\%)& 18 (95\%) & 0 (0\%)& 15 (79\%) & 0 (0\%)\\
\bottomrule
\end{tabular}
}
\end{sidewaystable*}

\begin{sidewaystable*}
\centering
\caption{CNOT optimization of a library of reversible functions with several CNOT circuits synthesis methods with the use of ancillary qubits. For each reversible function, "Original" reports some statistics (T-count, T-depth, CNOT count, Total depth) of the original circuit, "Tpar" reports the results of the circuit optimized solely by the Tpar algorithm, "Tpar + CNOT depth opt." reports the results of the circuit optimized with the Tpar algorithm and post-processed by our proposed depth optimization procedure.}
\label{table_results_ancilla}
\scalebox{.9}{
\begin{tabular}{lccc@{~}c@{~}c@{~}cc@{~}c@{~}c@{~}cc@{~}c@{~}c@{~}c@{~}c} \toprule
      \multirow{2}{*}{Function} & \multirow{2}{*}{\#$n$} & \multirow{2}{*}{\#Ancillae} & \multicolumn{4}{c}{Original} & \multicolumn{4}{c}{Tpar ($ \infty $ ancillae)}  & \multicolumn{5}{c}{Tpar + CNOT depth opt.} \\
      \cmidrule(lr){4-7} \cmidrule(lr){8-11} \cmidrule(lr){12-16}
               &   &  &  T-count & T-depth & CNOT-count & Depth &  T-count & T-depth & CNOT-count & Depth &  CNOT-count & \%Diff. & & Depth & \%Diff. \\
      \ \\
Adder\_8& $ 24$ & 19 & $399$ & $69$ & $466$ & $223$ & $215$ & $15$ & $1040$ & $380$ & $654$ & $ -37\% $ & & $171$ & $ -55\% $ \\
barenco\_tof\_10& $ 19$ & 9 & $224$ & $96$ & $224$ & $288$ & $100$ & $32$ & $396$ & $314$ & $216$ & $ -45\% $ & & $196$ & $ -38\% $ \\
barenco\_tof\_3& $ 5$ & 3 & $28$ & $12$ & $28$ & $36$ & $16$ & $4$ & $77$ & $58$ & $40$ & $ -48\% $ & & $29$ & $ -50\% $ \\
barenco\_tof\_4& $ 7$ & 4 & $56$ & $24$ & $56$ & $72$ & $28$ & $8$ & $124$ & $99$ & $66$ & $ -47\% $ & & $52$ & $ -47\% $ \\
barenco\_tof\_5& $ 9$ & 4 & $84$ & $36$ & $84$ & $108$ & $40$ & $12$ & $174$ & $144$ & $90$ & $ -48\% $ & & $76$ & $ -47\% $ \\
csla\_mux\_3& $ 15$ & 7 & $70$ & $21$ & $90$ & $67$ & $62$ & $4$ & $404$ & $211$ & $229$ & $ -43\% $ & & $57$ & $ -73\% $ \\
csum\_mux\_9& $ 30$ & 25 & $196$ & $18$ & $196$ & $59$ & $84$ & $3$ & $532$ & $166$ & $293$ & $ -45\% $ & & $55$ & $ -67\% $ \\
cycle\_17\_3& $ 35$ & 3 & $4739$ & $2001$ & $4742$ & $5974$ & $1944$ & $522$ & $6913$ & $5303$ & $4952$ & $ -28\% $ & & $3897$ & $ -27\% $ \\
GF($2^{10}$)\_mult& $ 30$ & 211 & $700$ & $108$ & $709$ & $290$ & $410$ & $2$ & $4597$ & $1095$ & $2643$ & $ -43\% $ & & $72$ & $ -93\% $ \\
GF($2^{16}$)\_mult& $ 48$ & 623 & $1792$ & $180$ & $1837$ & $489$ & $1040$ & $2$ & $12090$ & $2538$ & $8357$ & $ -31\% $ & & $101$ & $ -96\% $ \\
GF($2^{32}$)\_mult& $ 96$ & 2527 & $7168$ & $372$ & $7292$ & $1001$ & $4128$ & $2$ & $98111$ & $17660$ & $42134$ & $ -57\% $ & & $157$ & $ -99\% $ \\
GF($2^4$)\_mult& $ 12$ & 24 & $112$ & $36$ & $115$ & $99$ & $68$ & $2$ & $446$ & $147$ & $303$ & $ -32\% $ & & $46$ & $ -69\% $ \\
GF($2^5$)\_mult& $ 15$ & 52 & $175$ & $48$ & $179$ & $130$ & $115$ & $2$ & $863$ & $238$ & $487$ & $ -44\% $ & & $45$ & $ -81\% $ \\
GF($2^6$)\_mult& $ 18$ & 66 & $252$ & $60$ & $257$ & $163$ & $150$ & $2$ & $1009$ & $272$ & $792$ & $ -22\% $ & & $57$ & $ -79\% $ \\
GF($2^7$)\_mult& $ 21$ & 128 & $343$ & $72$ & $349$ & $195$ & $217$ & $2$ & $1855$ & $489$ & $1109$ & $ -40\% $ & & $58$ & $ -88\% $ \\
GF($2^8$)\_mult& $ 24$ & 128 & $448$ & $84$ & $469$ & $233$ & $264$ & $2$ & $2123$ & $548$ & $1590$ & $ -25\% $ & & $72$ & $ -87\% $ \\
GF($2^9$)\_mult& $ 27$ & 210 & $567$ & $96$ & $575$ & $258$ & $351$ & $2$ & $2709$ & $544$ & $2181$ & $ -19\% $ & & $69$ & $ -87\% $ \\
grover\_5& $ 9$ & 3 & $336$ & $144$ & $336$ & $457$ & $154$ & $44$ & $575$ & $535$ & $363$ & $ -37\% $ & & $360$ & $ -33\% $ \\
ham15-high& $ 20$ & 10 & $2457$ & $996$ & $2500$ & $3026$ & $1019$ & $262$ & $4285$ & $3295$ & $2656$ & $ -38\% $ & & $1864$ & $ -43\% $ \\
ham15-low& $ 17$ & 3 & $161$ & $69$ & $259$ & $263$ & $97$ & $20$ & $608$ & $414$ & $342$ & $ -44\% $ & & $185$ & $ -55\% $ \\
ham15-med& $ 17$ & 4 & $574$ & $240$ & $616$ & $750$ & $230$ & $53$ & $1087$ & $788$ & $636$ & $ -41\% $ & & $423$ & $ -46\% $ \\
hwb6& $ 7$ & 4 & $105$ & $45$ & $131$ & $152$ & $75$ & $13$ & $375$ & $285$ & $225$ & $ -40\% $ & & $139$ & $ -51\% $ \\
mod5\_4& $ 5$ & 4 & $28$ & $12$ & $32$ & $41$ & $16$ & $3$ & $77$ & $68$ & $45$ & $ -42\% $ & & $32$ & $ -53\% $ \\
mod\_adder\_1024& $ 28$ & 6 & $1995$ & $831$ & $2005$ & $2503$ & $1011$ & $230$ & $3871$ & $2614$ & $2683$ & $ -31\% $ & & $1675$ & $ -36\% $ \\
mod\_adder\_1048576& $ 58$ & 10 & $17290$ & $7292$ & $17310$ & $21807$ & $7298$ & $1879$ & $29315$ & $20154$ & $20808$ & $ -29\% $ & & $13875$ & $ -31\% $ \\
mod\_mult\_55& $ 9$ & 4 & $49$ & $15$ & $55$ & $50$ & $35$ & $4$ & $147$ & $87$ & $84$ & $ -43\% $ & & $47$ & $ -46\% $ \\
mod\_red\_21& $ 11$ & 4 & $119$ & $48$ & $122$ & $158$ & $73$ & $15$ & $308$ & $252$ & $163$ & $ -47\% $ & & $117$ & $ -54\% $ \\
qcla\_adder\_10& $ 36$ & 39 & $238$ & $24$ & $267$ & $73$ & $162$ & $6$ & $924$ & $219$ & $457$ & $ -51\% $ & & $65$ & $ -70\% $ \\
qcla\_com\_7& $ 24$ & 17 & $203$ & $27$ & $215$ & $81$ & $94$ & $7$ & $489$ & $172$ & $280$ & $ -43\% $ & & $79$ & $ -54\% $ \\
qcla\_mod\_7& $ 26$ & 22 & $413$ & $66$ & $441$ & $197$ & $237$ & $14$ & $1153$ & $352$ & $669$ & $ -42\% $ & & $150$ & $ -57\% $ \\
qft\_4& $ 5$ & 3 & $69$ & $48$ & $48$ & $142$ & $67$ & $38$ & $134$ & $203$ & $71$ & $ -47\% $ & & $133$ & $ -34\% $ \\
rc\_adder\_6& $ 14$ & 4 & $77$ & $33$ & $104$ & $104$ & $47$ & $11$ & $212$ & $133$ & $138$ & $ -35\% $ & & $81$ & $ -39\% $ \\
tof\_10& $ 19$ & 3 & $119$ & $51$ & $119$ & $153$ & $71$ & $17$ & $287$ & $221$ & $158$ & $ -45\% $ & & $111$ & $ -50\% $ \\
tof\_3& $ 5$ & 3 & $21$ & $9$ & $21$ & $27$ & $15$ & $3$ & $63$ & $56$ & $36$ & $ -43\% $ & & $28$ & $ -50\% $ \\
tof\_4& $ 7$ & 3 & $35$ & $15$ & $35$ & $45$ & $23$ & $5$ & $101$ & $91$ & $56$ & $ -45\% $ & & $39$ & $ -57\% $ \\
tof\_5& $ 9$ & 3 & $49$ & $21$ & $49$ & $63$ & $31$ & $7$ & $147$ & $126$ & $70$ & $ -52\% $ & & $51$ & $ -60\% $ \\
vbe\_adder\_3& $ 10$ & 6 & $70$ & $24$ & $80$ & $79$ & $24$ & $5$ & $116$ & $71$ & $75$ & $ -35\% $ & & $43$ & $ -39\% $ \\
\midrule
      Mean difference &&&&&&&&&&&&$-40.11\%$&&&$-57.86\%$\\
      Best savings &&&&&&&&&&&&$-57\%$&&&$-99\%$\\
      Worst savings &&&&&&&&&&&&$-19\%$&&&$-27\%$\\

      \bottomrule
      \end{tabular}
}
\end{sidewaystable*}

\begin{sidewaystable*}
\caption{Frequency of best performance of each algorithm during the
  optimization of reversible circuits with the use of ancillary qubits. For each algorithm, the first
  column gives the number of times it has returned the best result
  (possibly other algorithms returned circuits of same size). When
  available, the second column reports the number of times it was the
  only one to provide the best possible circuit. When there is only
  one column, it was never the only one to provide the best possible circuit.}
\centering
\scalebox{0.68}{
\begin{tabular}{lc@{~}c@{~}c@{~}cccccc@{~}cc@{~}cc@{~}cc@{~}cc@{~}cc@{~}c}
  \toprule
  \multirow{2}{*}{Function} & \multirow{2}{*}{\#$n$} & \multirow{2}{*}{\#Block} & \multirow{2}{*}{\#Direct} & \multirow{2}{*}{\makecell{\#CNOT \\ sub-circuits}} & \multicolumn{1}{c}{Kutin et al} & \multicolumn{1}{c}{DaCSynth} & \multicolumn{1}{c}{Greedy ($H_{\text{sum}}$, size)} & \multicolumn{1}{c}{Greedy ($h_{\text{prod}}$, size)} & \multicolumn{2}{c}{Greedy ($H_{\text{sum}}$)} & \multicolumn{2}{c}{Greedy ($h_{\text{prod}}$)} & \multicolumn{2}{c}{Greedy ($H_{\text{prod}}$)} & \multicolumn{1}{c}{LU + Greedy ($H_{\text{sum}}$)} & \multicolumn{1}{c}{LU + Greedy ($H_{\text{prod}}$)} \\
  \cmidrule(lr){6-6} \cmidrule(lr){7-7} \cmidrule(lr){8-8}
  \cmidrule(lr){9-9} \cmidrule(lr){9-9} \cmidrule(lr){10-11}
  \cmidrule(lr){12-13} \cmidrule(lr){14-15} \cmidrule(lr){16-16} \cmidrule(lr){17-17}
  & & & & & Best  & Best  & Best  & Best  & Best  & Only  & Best   & Only  & Best  & Only  & Best  & Best   \\
Adder\_8&43&0\%&49\% & 59 & 18 (31\%) &  5 (8\%) & 30 (51\%) & 28 (47\%) & 44 (75\%) & 3 (5\%)& 38 (64\%) & 1 (2\%)& 34 (58\%) & 0 (0\%)& 26 (44\%) & 26 (44\%) \\
barenco\_tof\_10&28&1\%&23\% & 101 & 57 (56\%) &  39 (39\%) & 85 (84\%) & 85 (84\%) & 87 (86\%) & 0 (0\%)& 87 (86\%) & 0 (0\%)& 86 (85\%) & 0 (0\%)& 70 (69\%) & 70 (69\%) \\
barenco\_tof\_3&8&0\%&70\% & 24 & 8 (33\%) &  4 (17\%) & 10 (42\%) & 9 (38\%) & 10 (42\%) & 0 (0\%)& 10 (42\%) & 0 (0\%)& 10 (42\%) & 0 (0\%)& 9 (38\%) & 8 (33\%) \\
barenco\_tof\_4&11&5\%&38\% & 35 & 16 (46\%) &  9 (26\%) & 19 (54\%) & 19 (54\%) & 21 (60\%) & 0 (0\%)& 21 (60\%) & 0 (0\%)& 19 (54\%) & 0 (0\%)& 16 (46\%) & 16 (46\%) \\
barenco\_tof\_5&13&0\%&31\% & 46 & 25 (54\%) &  14 (30\%) & 30 (65\%) & 30 (65\%) & 32 (70\%) & 0 (0\%)& 32 (70\%) & 0 (0\%)& 30 (65\%) & 0 (0\%)& 25 (54\%) & 25 (54\%) \\
csla\_mux\_3&22&0\%&27\% & 11 & 1 (9\%) &  1 (9\%) & 3 (27\%) & 3 (27\%) & 6 (55\%) & 0 (0\%)& 10 (91\%) & 5 (45\%)& 5 (45\%) & 0 (0\%)& 1 (9\%) & 2 (18\%) \\
csum\_mux\_9&55&0\%&0\% & 23 & 1 (4\%) &  3 (13\%) & 3 (13\%) & 3 (13\%) & 9 (39\%) & 2 (9\%)& 6 (26\%) & 0 (0\%)& 5 (22\%) & 0 (0\%)& 2 (9\%) & 2 (9\%) \\
cycle\_17\_3&38&0\%&5\% & 1410 & 582 (41\%) &  450 (32\%) & 892 (63\%) & 853 (60\%) & 1374 (97\%) & 1 (0\%)& 1169 (83\%) & 15 (1\%)& 1273 (90\%) & 1 (0\%)& 826 (59\%) & 817 (58\%) \\
GF($2^{10}$)\_mult&241&67\%&0\% & 42 & 1 (2\%) &  1 (2\%) & 2 (5\%) & 2 (5\%) & 14 (33\%) & 4 (10\%)& 16 (38\%) & 8 (19\%)& 15 (36\%) & 3 (7\%)& 2 (5\%) & 2 (5\%) \\
GF($2^{16}$)\_mult&671&83\%&0\% & 58 & 1 (2\%) &  1 (2\%) & 1 (2\%) & 1 (2\%) & 19 (33\%) & 10 (17\%)& 22 (38\%) & 16 (28\%)& 18 (31\%) & 5 (9\%)& 1 (2\%) & 1 (2\%) \\
GF($2^{32}$)\_mult&2623&100\%&0\% & 100 & 1 (1\%) &  1 (1\%) & 1 (1\%) & 1 (1\%) & 15 (15\%) & 3 (3\%)& 56 (56\%) & 48 (48\%)& 35 (35\%) & 18 (18\%)& 1 (1\%) & 1 (1\%) \\
GF($2^4$)\_mult&36&17\%&0\% & 12 & 2 (17\%) &  2 (17\%) & 3 (25\%) & 2 (17\%) & 10 (83\%) & 1 (8\%)& 10 (83\%) & 2 (17\%)& 6 (50\%) & 0 (0\%)& 2 (17\%) & 2 (17\%) \\
GF($2^5$)\_mult&67&33\%&33\% & 10 & 2 (20\%) &  2 (20\%) & 4 (40\%) & 3 (30\%) & 7 (70\%) & 0 (0\%)& 9 (90\%) & 3 (30\%)& 5 (50\%) & 0 (0\%)& 2 (20\%) & 2 (20\%) \\
GF($2^6$)\_mult&84&50\%&17\% & 28 & 2 (7\%) &  2 (7\%) & 1 (4\%) & 1 (4\%) & 9 (32\%) & 2 (7\%)& 8 (29\%) & 4 (14\%)& 7 (25\%) & 1 (4\%)& 3 (11\%) & 3 (11\%) \\
GF($2^7$)\_mult&149&33\%&33\% & 12 & 1 (8\%) &  1 (8\%) & 1 (8\%) & 1 (8\%) & 4 (33\%) & 0 (0\%)& 11 (92\%) & 5 (42\%)& 7 (58\%) & 0 (0\%)& 1 (8\%) & 1 (8\%) \\
GF($2^8$)\_mult&152&33\%&17\% & 33 & 1 (3\%) &  1 (3\%) & 2 (6\%) & 1 (3\%) & 12 (36\%) & 2 (6\%)& 13 (39\%) & 5 (15\%)& 11 (33\%) & 0 (0\%)& 1 (3\%) & 1 (3\%) \\
GF($2^9$)\_mult&237&83\%&0\% & 40 & 1 (2\%) &  1 (2\%) & 4 (10\%) & 4 (10\%) & 17 (42\%) & 4 (10\%)& 14 (35\%) & 4 (10\%)& 17 (42\%) & 2 (5\%)& 2 (5\%) & 2 (5\%) \\
grover\_5&12&0\%&9\% & 130 & 66 (51\%) &  57 (44\%) & 115 (88\%) & 106 (82\%) & 116 (89\%) & 1 (1\%)& 115 (88\%) & 0 (0\%)& 114 (88\%) & 0 (0\%)& 88 (68\%) & 87 (67\%) \\
ham15-high&30&1\%&22\% & 745 & 395 (53\%) &  263 (35\%) & 628 (84\%) & 602 (81\%) & 737 (99\%) & 21 (3\%)& 684 (92\%) & 6 (1\%)& 641 (86\%) & 1 (0\%)& 524 (70\%) & 519 (70\%) \\
ham15-low&20&2\%&34\% & 56 & 18 (32\%) &  7 (12\%) & 47 (84\%) & 44 (79\%) & 51 (91\%) & 2 (4\%)& 54 (96\%) & 5 (9\%)& 38 (68\%) & 0 (0\%)& 25 (45\%) & 24 (43\%) \\
ham15-med&21&0\%&35\% & 151 & 73 (48\%) &  41 (27\%) & 134 (89\%) & 124 (82\%) & 150 (99\%) & 3 (2\%)& 141 (93\%) & 1 (1\%)& 121 (80\%) & 0 (0\%)& 95 (63\%) & 94 (62\%) \\
hwb6&11&0\%&30\% & 51 & 15 (29\%) &  10 (20\%) & 34 (67\%) & 28 (55\%) & 35 (69\%) & 1 (2\%)& 33 (65\%) & 2 (4\%)& 31 (61\%) & 0 (0\%)& 19 (37\%) & 17 (33\%) \\
mod5\_4&9&0\%&40\% & 10 & 8 (80\%) &  7 (70\%) & 9 (90\%) & 8 (80\%) & 10 (100\%) & 1 (10\%)& 9 (90\%) & 0 (0\%)& 8 (80\%) & 0 (0\%)& 8 (80\%) & 8 (80\%) \\
mod\_adder\_1024&34&0\%&6\% & 687 & 360 (52\%) &  242 (35\%) & 487 (71\%) & 474 (69\%) & 663 (97\%) & 0 (0\%)& 633 (92\%) & 3 (0\%)& 611 (89\%) & 1 (0\%)& 449 (65\%) & 436 (63\%) \\
mod\_adder\_1048576&68&0\%&1\% & 5459 & 2644 (48\%) &  1888 (35\%) & 3607 (66\%) & 3537 (65\%) & 5268 (97\%) & 6 (0\%)& 4727 (87\%) & 71 (1\%)& 4903 (90\%) & 4 (0\%)& 3367 (62\%) & 3350 (61\%) \\
mod\_mult\_55&109&9\%&18\% & 11 & 5 (45\%) &  2 (18\%) & 7 (64\%) & 5 (45\%) & 10 (91\%) & 1 (9\%)& 8 (73\%) & 0 (0\%)& 10 (91\%) & 0 (0\%)& 5 (45\%) & 4 (36\%) \\
mod\_red\_21&15&0\%&44\% & 50 & 21 (42\%) &  7 (14\%) & 34 (68\%) & 32 (64\%) & 36 (72\%) & 0 (0\%)& 33 (66\%) & 0 (0\%)& 30 (60\%) & 0 (0\%)& 27 (54\%) & 28 (56\%) \\
qcla\_adder\_10&75&6\%&24\% & 31 & 5 (16\%) &  4 (13\%) & 9 (29\%) & 7 (23\%) & 16 (52\%) & 1 (3\%)& 12 (39\%) & 0 (0\%)& 15 (48\%) & 0 (0\%)& 7 (23\%) & 7 (23\%) \\
qcla\_com\_7&41&0\%&35\% & 34 & 8 (24\%) &  4 (12\%) & 15 (44\%) & 14 (41\%) & 17 (50\%) & 1 (3\%)& 16 (47\%) & 0 (0\%)& 16 (47\%) & 2 (6\%)& 9 (26\%) & 9 (26\%) \\
qcla\_mod\_7&48&8\%&16\% & 51 & 9 (18\%) &  5 (10\%) & 26 (51\%) & 17 (33\%) & 32 (63\%) & 3 (6\%)& 24 (47\%) & 0 (0\%)& 26 (51\%) & 3 (6\%)& 13 (25\%) & 13 (25\%) \\
qft\_4&105&0\%&43\% & 23 & 15 (65\%) &  10 (43\%) & 11 (48\%) & 11 (48\%) & 23 (100\%) & 0 (0\%)& 23 (100\%) & 0 (0\%)& 20 (87\%) & 0 (0\%)& 9 (39\%) & 9 (39\%) \\
rc\_adder\_6&18&0\%&49\% & 51 & 25 (49\%) &  21 (41\%) & 33 (65\%) & 32 (63\%) & 37 (73\%) & 0 (0\%)& 35 (69\%) & 0 (0\%)& 34 (67\%) & 0 (0\%)& 29 (57\%) & 29 (57\%) \\
tof\_10&22&0\%&48\% & 56 & 24 (43\%) &  17 (30\%) & 40 (71\%) & 38 (68\%) & 42 (75\%) & 0 (0\%)& 41 (73\%) & 0 (0\%)& 39 (70\%) & 0 (0\%)& 34 (61\%) & 34 (61\%) \\
tof\_3&105&0\%&50\% & 8 & 4 (50\%) &  2 (25\%) & 2 (25\%) & 2 (25\%) & 8 (100\%) & 0 (0\%)& 7 (88\%) & 0 (0\%)& 5 (62\%) & 0 (0\%)& 2 (25\%) & 2 (25\%) \\
tof\_4&10&0\%&58\% & 26 & 2 (8\%) &  2 (8\%) & 10 (38\%) & 9 (35\%) & 12 (46\%) & 0 (0\%)& 11 (42\%) & 0 (0\%)& 8 (31\%) & 0 (0\%)& 7 (27\%) & 7 (27\%) \\
tof\_5&12&0\%&59\% & 31 & 9 (29\%) &  5 (16\%) & 15 (48\%) & 13 (42\%) & 17 (55\%) & 0 (0\%)& 16 (52\%) & 0 (0\%)& 15 (48\%) & 0 (0\%)& 14 (45\%) & 15 (48\%) \\
vbe\_adder\_3&16&7\%&7\% & 29 & 6 (21\%) &  4 (14\%) & 13 (45\%) & 12 (41\%) & 14 (48\%) & 1 (3\%)& 14 (48\%) & 1 (3\%)& 12 (41\%) & 0 (0\%)& 9 (31\%) & 8 (28\%) \\

\bottomrule
\end{tabular}
}
\label{frequency_ancilla}
\end{sidewaystable*}

\section*{Acknowledgment}
This work was supported in part by the French National Research Agency
(ANR) under the research project SoftQPRO ANR-17-CE25-0009-02,
and by the DGE of the French Ministry of Industry under the research
project PIA-GDN/QuantEx P163746-484124.


\appendix

\section{Kutin et al's algorithm for LNN connectivity: presentation and extension to full qubit connectivity} \label{appendix::kutin}

Kutin \textit{et al.} gave several constructions of specific linear reversible operations for the LNN architecture: addition, swap, permutation, generic linear reversible operator \cite{DBLP:journals/cjtcs/KutinMS07}. They focused on the shallowest way to do it. For a generic linear reversible operator, they relied on their construction for reversing the qubits, i.e., the image of an $n$-qubit state $\ket{x_1x_2...x_n}$ is $\ket{x_nx_{n-1}...x_1}$. This construction is a \emph{sorting network} and contains only SWAP gates. The network, as a SWAP circuit, is of depth $n$. An example with $7$ qubits is given in Fig~\ref{preamble::sorting}. Then they considered the same sorting network but with \emph{boxes} replacing the SWAP gates. Each box, acting on $2$ qubits, can perform one of the following operations: 
\begin{itemize}
  \item $(u,v) \to (u,v)$, requiring $0$ CNOT,
  \item $(u,v) \to (u,u\oplus v)$, requiring $1$ CNOT,
  \item $(u,v) \to (u \oplus v,v)$, requiring $1$ CNOT,
  \item $(u,v) \to (v,u\oplus v)$, requiring $2$ CNOT,
  \item $(u,v) \to (u\oplus v,u)$, requiring $2$ CNOT,
  \item $(u,v) \to (v,u)$, requiring $3$ CNOT.
\end{itemize}

\begin{figure}
\centering
\includegraphics[scale=0.6]{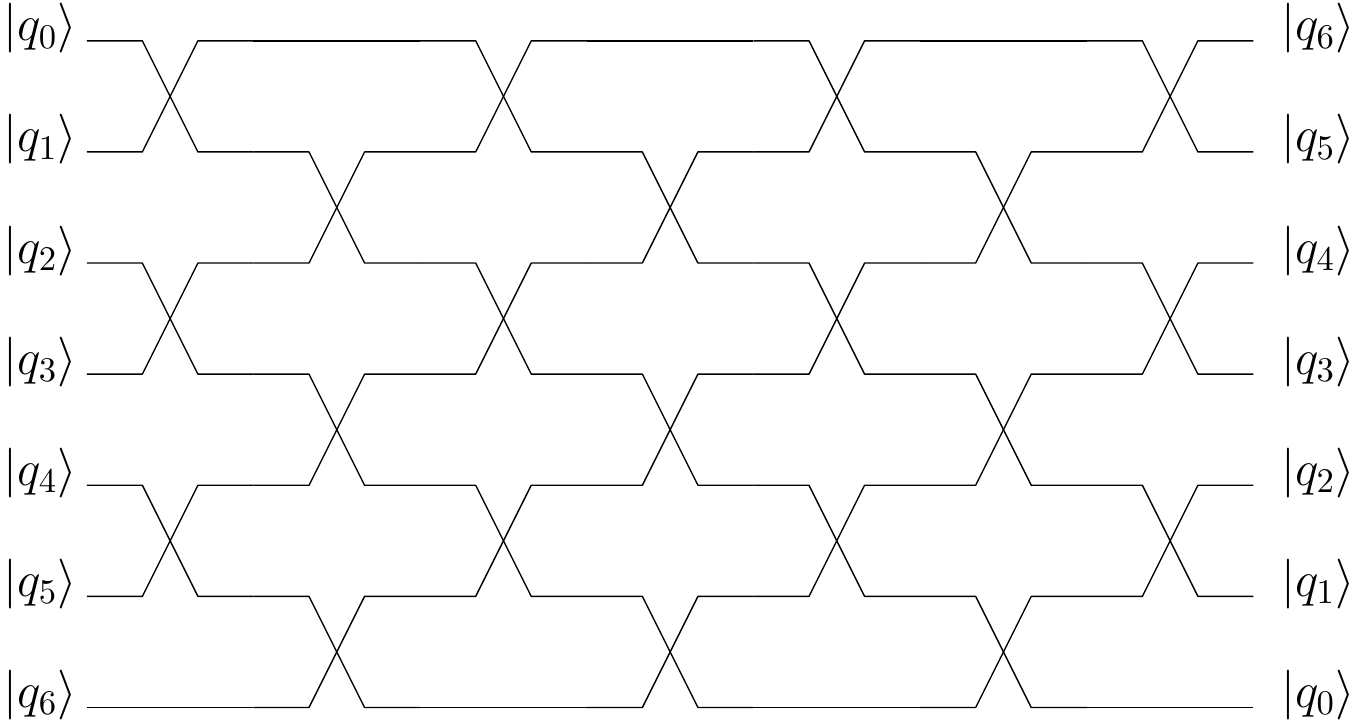}
\caption[Sorting network for $7$ qubits.]{Sorting network for $7$ qubits. As a SWAP circuit, the depth of the circuit is $n$. Replacing each SWAP by a box gives a skeleton circuit for the synthesis of triangular linear reversible operators.}
\label{preamble::sorting}
\end{figure}

Kutin \textit{et al.} \cite{DBLP:journals/cjtcs/KutinMS07} proved that a sorting network made of boxes can transform any operator into a northwest triangular one. Moreover, for each box, only the 
state of one of the two output qubits needs to be fixed after applying the box. This means that we can always choose a box that needs at most $2$ CNOTs to be implemented. Consequently, the total depth of the sorting network is $2n$. Finally, they showed how to synthesize a northwest triangular operator with a similar sorting network 
except that in this case for each box the states of the two output qubits need to be fixed. Therefore we may need at most $3$ CNOTs for some boxes (if we only need to swap the qubits) and the depth of this second part is upper bounded by $3n$. Overall this gives a generic method for synthesizing any linear reversible operator for the LNN architecture in depth at most $5n$. To our knowledge, this is the best result in the literature for the case of restricted connectivity. This result can only be improved by a constant factor as Kutin \textit{et al.} also showed that some operators need at least circuits with depth $2n$ to be implemented. So the best possible synthesis method for the LNN architecture should provide circuits of depth comprised between $2n$ and $5n$. For other architectures, the bounds are not clear. Obviously, if an architecture contains a Hamiltonian path in it then one can apply the algorithm for the LNN case, giving an upper bound of $5n$ for the depth. To our knowledge, lower bounds are not known but Maslov computed lower bounds for $2$ simplified models in the case where each qubit has $k$ neighbors \cite{maslov2007linear}. The first model is the case where we have to execute every gate given by the Gaussian elimination algorithm in a given order; the second model is less restrictive as we have to execute every gate but we assume that they all commute. In both cases, the depth is lower bounded linearly in $n$. 

\subsection*{Extension to Full Qubit Connectivity} \label{cnot::soa_depth}

Although it was not done in their paper, the algorithm proposed by Kutin \textit{et al.} \cite{DBLP:journals/cjtcs/KutinMS07} can be extended to the full connectivity case: this is what we show in this paragraph. To our knowledge, such an extension has never been proposed in the literature.

In the original Kutin \textit{et al.}'s algorithm
\cite{DBLP:journals/cjtcs/KutinMS07}, each box corresponds to an
interaction between a pair of qubits and it can be decomposed into two
parts: first, we execute the interaction strictly speaking between the
two qubits with a CNOT gate, secondly, we move the qubits in the
hardware by swapping them. If we consider that the connectivity is
full then we do not need to move the qubits anymore in the
hardware. This means that we can replace each box by a CNOT gate and
we get rid of the SWAP gates. We end with a new skeleton circuit that
is functionally equivalent to the one given by Kutin \textit{et al.}
except that each box is now a single CNOT. The skeleton circuit from
Kutin \textit{et al.}, as a box-based circuit, is of depth
$2n$. Therefore our new skeleton, as a CNOT-based circuit, is also of
depth $2n$.

\smallskip
To our knowledge, this was the best result until the
asymptotically optimal algorithm proposed recently in
\cite{DBLP:conf/soda/JiangSTW0Z20}. The pseudo-code of this new
algorithm is given Algorithm~\ref{PseudoCode_Kutin}. For simplicity we
only show the case for a lower triangular operator, the generalization
to any operator is done via an LU decomposition \cite{GVL96} stating that
\( A = PLU\)
where $A$ is the operator to synthesize, $P$ is a permutation matrix, and $L$, resp. $U$, are lower, resp. upper, triangular operators. With full qubit connectivity, a permutation can be implemented with a circuit of constant depth $6$ \cite{DBLP:journals/siamcomp/MooreN01}. Each triangular operator can be synthesized with a circuit of depth $n$, leading to a total depth of $2n+6$ for the synthesis of an arbitrary operator. Given that we do not move the qubits anymore, most of the algorithm consists in tracking what would be the positions of the qubits in the hardware to determine which interactions need to be done at a given time step. Then it is easy to decide if, for a given pair of qubits $(i,j)$, a CNOT gate needs to be added. If the operator is lower triangular, we only have to decide if we add the CNOT $(i \to j), i<j$ as we must not ruin the triangular structure by adding a CNOT $(j \to i)$. Then if the $i$-th component of the $j$-th row is $1$ then add a CNOT $(i \to j)$. The reason why it works is not straightforward: we have to note that when we decide to apply or not a CNOT $(i \to j)$, either the components $k<i$ have already been treated for qubit $i$ so it cannot modify the components of qubit $j$, or such components have not been treated on both qubits, so modifying them on qubit $j$ is not a problem as they will be zeroed later in the algorithm.

\begin{algorithm}
\caption{Adaptation of Kutin \textit{et al.}'s algorithm \cite{DBLP:journals/cjtcs/KutinMS07} on a triangular operator $L$ for a full qubit connectivity.}
\label{PseudoCode_Kutin}
\scalebox{.9}{\begin{minipage}{.52\textwidth}
\begin{algorithmic} 
\REQUIRE $n \geq 0, \; \; L \in \mathbb{F}_2^{n \times n}$ triangular
\ENSURE $C$ is a CNOT-circuit implementing $L$\,with\,depth at most $n$
\STATE
\STATE $C \leftarrow [\;]$
\STATE perm $\leftarrow \llbracket 1,n \rrbracket$
\FOR{$j$ = $1$ to $n$}
\IF{$j \equiv 1 [2]$}
\STATE start $\leftarrow 1$
\ELSE
\STATE start $\leftarrow 2$
\ENDIF
\WHILE{start $< n$}
\IF{$L[\text{perm[start+1], perm[start]}] = 1$}
\STATE C.append(CNOT(perm[start], perm[start+1]))
\ENDIF
\STATE perm[start], perm[start+1] $\leftarrow$ perm[start+1], perm[start]
\STATE start $\leftarrow$ start + 2
\ENDWHILE
\ENDFOR
\RETURN reverse(C)
\end{algorithmic}
\end{minipage}}
\end{algorithm}

\section{Jiang et al.'s algorithm} \label{appendix::jiang}

\begin{figure*}[!tb]
\centering
\includegraphics[scale=.9]{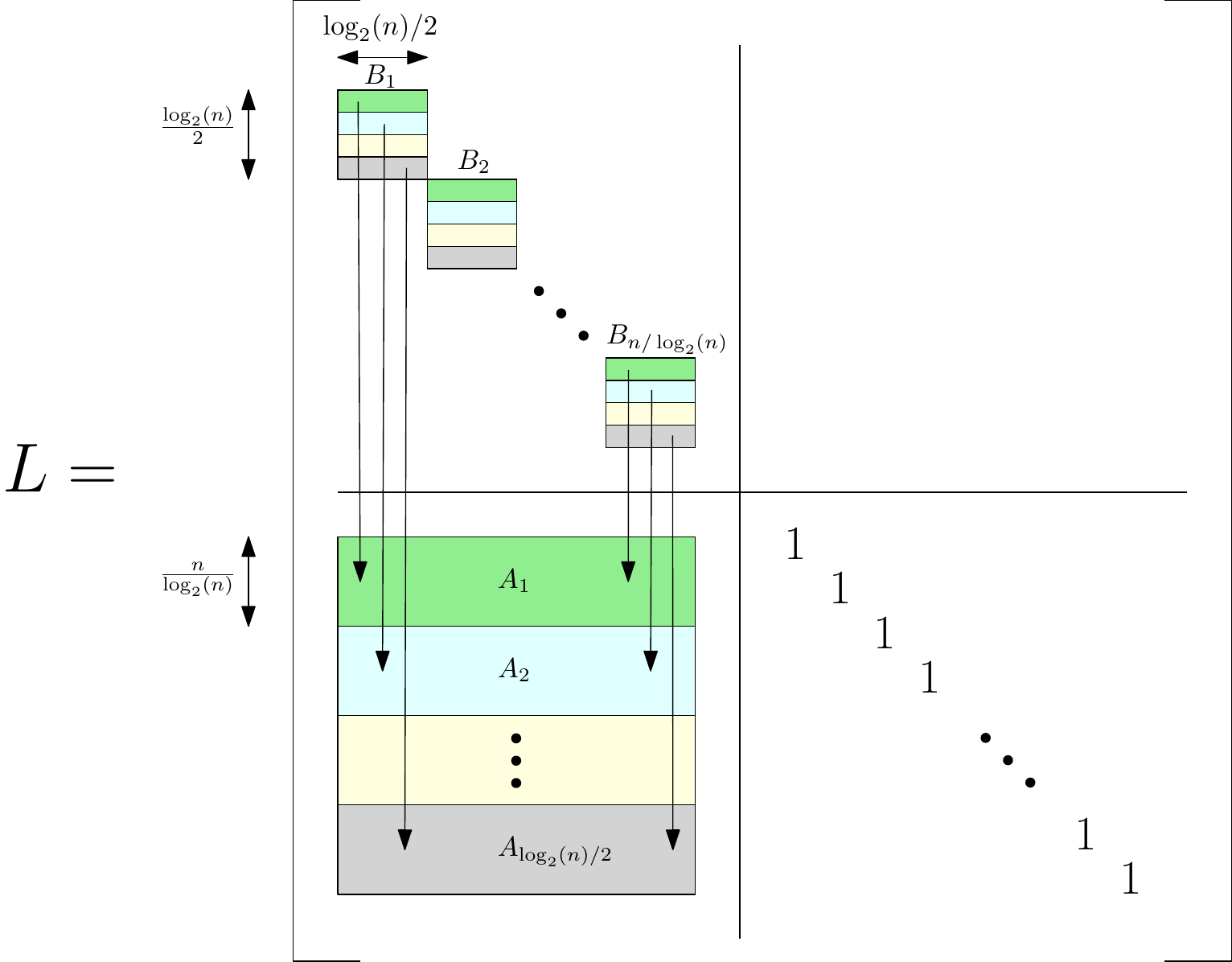}
\caption{Block structure of the triangular operator $L$ for Jiang \emph{et al.}'s algorithm.}
\label{jiang}
\end{figure*}

The authors \cite{DBLP:conf/soda/JiangSTW0Z20} propose an algorithm based on the LU decomposition and a divide-and-conquer approach.
The proof of the optimal depth complexity is quite hard to summarize but the principle of the algorithm is simpler so we give a brief description of it.
First, the algorithm starts with an LU decomposition $A = PLU$. Again $P$ can be synthesized in constant depth $6$, therefore we only need to treat the triangular case. We illustrate with the lower triangular case. 
The synthesis of $L$ consists in a divide-and-conquer algorithm, the operator $L$ is decomposed as 
\[ L = \begin{pmatrix} L_{\floor*{n/2}} &  \\ A & L_{\ceil*{n/2}} \end{pmatrix}\]
where $L_{\floor*{n/2}}$ and $L_{\ceil*{n/2}}$ are triangular
operators of  respective sizes $\floor*{n/2}$ and $\ceil*{n/2}$ and $A$ is any Boolean matrix of size 
$\floor*{n/2} \times \ceil*{n/2}$. The algorithm initially synthesizes in parallel both triangular suboperators by applying recursively the algorithm. Then we are left with the operator
\[ L' = \begin{pmatrix} I &  \\ A' & I \end{pmatrix} \]
to synthesize.

\smallskip
This is done by considering the following blocks in $L$: 
\begin{itemize}
  \item the northwest identity operator is seen as a block diagonal operator with $n/\log_2(n)$ blocks of size $\log_2(n)/2$, noted $B_1, ..., B_{n/\log_2(n)}$,
  \item $A$ is divided into $\log_2(n)/2$ blocks of $n/\log_2(n)$ rows, noted $A_1, ..., A_{\log_2(n)/2}$. For simplicity we consider each $A_i$ as a matrix $C_i \in (F_2^{\log_2(n)/2})^{\frac{n}{\log_2(n)} \times \frac{n}{\log_2(n)}}$, i.e., we see $A_i$ as a $\frac{n}{\log_2(n)} \times \frac{n}{\log_2(n)}$ matrix with elements from $F_2^{\log_2(n)/2}$.
\end{itemize}

The specific structure of $L$ is summarized in Fig~\ref{jiang}. The synthesis of $L$ consists in successive applications of two stages of row operations: 
\begin{enumerate}
  \item row operations on the $B_i$'s such that specific words of $\log_2(n)/2$ bits appear on each row,
  \item row operations between the $B_i$'s and the $A_i$'s to zero words of $\log_2(n)/2$ bits of $A$.
\end{enumerate}

More precisely, the $k$-th row of a $B_i$ has to zero the $i$-th
column of $C_k$. For that a sequence of operators on $\log_2(n)/2$
qubits is computed such that the property "Every word on $\log_2(n)/2$
bits appears on each row of the $B_i$'s" is verified. Such sequence of
operators is called a row traversal sequence. The operators of the row
traversal sequence are computed in parallel on each $B_i$ via the row
operations during Stage 1. Once we have the desired operator on each
$B_i$, we need to compute the appropriate row operations of Stage
2. Each row of a $B_i$ act on a specific $C_k$ so the corresponding
row operations can be done in parallel. So all we need is to see how
to coordinate the row operations acting on the same $C_k$. For
simplicity consider the case $k=1$, i.e., all the first rows of each
$B_i$ are used to zero $C_1$. Given that all the $B_i$'s are identical,
the choice of applying a row operation from block $B_i$ to the $k$-th
row of $C_1$ is: is $C[k,i]$ equal to $B_*[1,:]$, where the index on
$B$ has been omitted to emphasize that all the $B_i$'s are equal.

Therefore the matrix $C_1$ can be seen as the adjacency matrix $P$ of a bipartite graph $G$ where $P[k,i]=1$ if $C_1[k,i] = B_*[1,:]$. A sequence of parallel row operations between the $B_i$'s and $C_1$ corresponds to a matching in $G$ and a "good" sequence of parallel row operations is given by a matching decomposition of $G$. A central theorem that we will also use in our own work is the following: if the maximum number of $1$ in a row or a column of $P$ is $p$ then there exists a decomposition of $G$ into $p$ matchings, i.e., a sequence of $p$ parallel row operations is necessary to zero all the entries of $C$ equal to $B_*[1,:]$. Given that each word appears on each row of the $B_i$'s we are ensured that $A$ will be zero at the end of the algorithm. 

\medskip
To conclude, we need to assume that $A$ is sufficiently random, if it is not the case one can decompose $A = A' \oplus {A''}$ with $A', {A''}$ sufficiently random and do the process two times, the first time for adding $A'$ and the second time for adding ${A''}$. 
The depth $d(n)$ for the synthesis of one triangular operator is
therefore given by the equation shown in Figure~\ref{fig:equation}.
\begin{figure*}[!bt]
\[ d(n) = d(n/2) + 2 \times \text{ length row traversal sequence }
  \times \left( \underbrace{d(\log_2(n))}_{\text{synthesize the
        operator $B_i$}} + \text{ size matching decomposition}
  \right)\]
\caption{Depth $d(n)$ for the synthesis of one triangular operator}
\label{fig:equation}
\end{figure*}

\paragraph{}
Finally, the authors \cite{DBLP:conf/soda/JiangSTW0Z20} have shown
that the length of the row traversal sequence is
$\mathcal{O}(\sqrt{n})$ and if $A$ is sufficiently random at each
iteration the matching decomposition is of size
$\mathcal{O}(\sqrt{n}/\log_2(n))$. Therefore
\begin{align*} d(n) & = d(n/2) + \mathcal{O}(\sqrt{n}) \times (
                      d(\log_2(n)) + \mathcal{O}(\sqrt{n}/\log_2(n)))
  \\ 
                    & = \alpha(n/\log_2(n)) + \beta(\sqrt{n}/\log_2(n))) \end{align*}
hence the result.

\medskip
Let us now compute an estimation for $\alpha, \beta$. As a
divide-and-conquer framework, the authors have derived a recursive
formula in the case of triangular operators. Noting $d(n)$ for the
depth we have
\[ 
d(n) \leq d(n/2) + 2 \times \mathcal{O}(\sqrt{n}) \times (
  \mathcal{O}(\log_2(n)) + \mathcal{O}(\sqrt{n}/\log_2(n))).
\] 

Note that we think there is a typo in their formula, the term
$\mathcal{O}(\log_2(n))$ being $\mathcal{O}({\log^3}_2(n))$ in their
paper. This term corresponds to the synthesis of an operator of size
$\log_2(n)/2$. Assuming that we use the best algorithm, i.e., the
adaptation of Kutin \textit{et al.'}s algorithm
\cite{DBLP:journals/cjtcs/KutinMS07} we proposed in
Appendix~\ref{cnot::soa_depth}, each of these operators can be
synthesized with a circuit of depth at most
$2 \times \log_2(n)/2 = \log_2(n).$ We may have missed something but
this improves the real complexity so we keep our proposed
modification. The second term $\mathcal{O}(\sqrt{n}/\log_2(n)))$
corresponds to the matching decomposition of a graph and the authors
showed that the leading coefficient is $\sqrt{e}$. The third term,
$\mathcal{O}(\sqrt{n})$, corresponds to the length of the
row-traversal sequence on $k$ qubits that gives a sequence of
$k$-qubit operators such that for any bitstring of size $k$ and any
integer $j \in \llbracket 1,k \rrbracket$, there is an operator in the
sequence whose $j$-th row equals the bitstring. The authors proved
that there exists a row-traversal sequence on $k$ qubits of length
$3 \times 2^{k-1} - k + 1$. Here we have $k = \log_2(n)/2$ and
$\mathcal{O}(\sqrt{n}) = 3/2 \times \sqrt{n}$.

We therefore finally get
\begin{align*} d(n) &\leq d(n/2) + 3 \times \sqrt{n} \times \left( \log_2(n) + \sqrt{ne}/\log_2(n) \right) \\ &\leq d(n/2) + 3\sqrt{n}\log_2(n) + 3\frac{n\sqrt{e}}{\log_2(n)}\end{align*}
and
\[d(n) \leq 3 \times \left( \sum_{j=0}^{\log_2(n)-1} \sqrt{\frac{n}{2^j}} \log_2\left(\frac{n}{2^j}\right) + \frac{\sqrt{e}\frac{n}{2^j}}{\log_2\left(\frac{n}{2^j}\right)}  \right). \]
After simplification we have 
\[ d(n) \leq 3 \times \left(2\sqrt{e} \frac{n}{\log_2(n)} + 3.3 \sqrt{n}\log_2(n) \right)\]
and we have to do it for the two triangular operators given by the LU decomposition. Overall 
\[ \text{depth \cite{DBLP:conf/soda/JiangSTW0Z20}} \leq 20 \left(\frac{n}{\log_2(n)} + \sqrt{n}\log_2(n) \right).\]

Although it is only an upper bound, in practice there is little simplification one can make when synthesizing a specific operator: the row-traversal sequence still needs to be synthesized entirely and the matching decomposition is done on random graphs so we cannot expect the exact complexity to be that lower compared to the upper bound.

\end{document}